\documentclass[letterpaper,11pt]{article}
\usepackage[utf8]{inputenc}
\usepackage[namelimits]{amsmath} 
\usepackage{forest}
\usepackage{fullpage}
\usepackage{color}
\usepackage{amsfonts}
\usepackage{titling}
\usepackage[left=1.0in,right=1.0in,top=1.0in,bottom=1.0in]{geometry}
\definecolor{Blue}{rgb}{0.1,0.1,0.8}
\usepackage{hyperref}
\hypersetup{
    linktocpage=true,
    colorlinks=true,				
    linkcolor=Blue,				
    citecolor=Blue,				
    urlcolor=Blue,			
}
\usepackage{makecell}
\usepackage{tablefootnote}
\usepackage{graphicx}  
\usepackage{hyperref}
\usepackage[ruled]{algorithm2e}
\usepackage{mathrsfs}
\usepackage[english]{babel}
\usepackage{amssymb}
\usepackage{multirow}
\usepackage{soul}
\usepackage{enumitem}

\usepackage{amsthm}
\usepackage{amsmath}
\usepackage{verbatim}
\usepackage{amssymb,tikz}
\usepackage{blindtext}
\usepackage{booktabs}
\newtheorem{theorem}{Theorem}[section]
\newtheorem{corollary}{Corollary}[section]
\newtheorem{lemma}[theorem]{Lemma}
\newtheorem{example}{Example}
\newtheorem{definition}{Definition}
 
\newcommand{\dom}{\mathbf{dom}}
\newcommand{\argmax}{\mathop{\mathrm{arg\,max}}}

\newcommand{\I}{\mathbf{I}}
\newcommand{\x}{\mathbf{x}}
\newcommand{\s}{\mathbf{s}}
\newcommand{\y}{\mathbf{y}}

\newcommand{\out}{\mathbf{o}}

\newcommand{\adom}{\mathbb{Z}^*}
\newcommand{\audom}{\mathbb{Z}^+}
\newcommand{\fdom}{\mathbb{Z}}

\title{A Nearly Instance-optimal Differentially Private Mechanism for Conjunctive Queries}

\author{
Wei Dong
\and
Ke Yi
}
\date{Hong Kong University of Science and Technology}

\begin{document}
\maketitle 

\begin{abstract}
Releasing the result size of conjunctive queries and graph pattern queries under differential privacy (DP) has received considerable attention in the literature, but existing solutions do not offer any optimality guarantees.  We provide the first DP mechanism for this problem with a fairly strong notion of optimality, which can be considered as a natural relaxation of instance-optimality to a constant. 
\end{abstract}

\section{Introduction}
\label{sec:intro}

The bulk of work on \textit{differential privacy (DP)} has been devoted to counting queries \cite{dwork2014algorithmic,vadhan2017complexity}. For a query $q$, let $q(\mathbf{I})$ be the results of evaluating $q$ on database instance $\mathbf{I}$, and let $|q(\mathbf{I})|$ be the cardinality of $q(\mathbf{I})$.  A DP mechanism $\mathcal{M}_q(\I)$ for a counting query $q$ aims to release $|q(\I)|$ masked with noise so as to satisfy the DP requirement.  Formally, $\mathcal{M}_q(\cdot)$ is \textit{$\varepsilon$-differentially private} if 
\begin{equation}
\label{eq:DPdef}
\Pr[\mathcal{M}_q(\mathbf{I}) = y ]\le e^\varepsilon\Pr[\mathcal{M}_q(\mathbf{I}') = y] 
\end{equation}
for any $y$ and any pair of neighboring instances $\mathbf{I},\mathbf{I}'$, i.e., $d(\mathbf{I}, \mathbf{I}')=1$. Here, $d(\mathbf{I}, \mathbf{I'})$ denotes the minimum number of changes to turn $\mathbf{I}$ into $\mathbf{I}'$. 
DP policies may vary depending on what constitutes a ``change''.  In strict DP~\cite{johnson2018towards,kotsogiannis2019privatesql}, a change can be inserting a tuple, deleting a tuple, or substituting a tuple by another, while in a relaxed version, only substitutions are allowed~\cite{kifer2011no,balle2018privacy}. The latter is more relaxed since there are less neighboring instances that must satisfy \eqref{eq:DPdef}.  In particular, neighboring instances must have the same size, so a DP mechanism may use the instance size $N=|\I|$ directly; on the other hand, $N$ must be protected in strict DP. All the positive results in this paper hold for strict DP, while negative results hold even for the relaxed version.  
The privacy parameter $\varepsilon$ is usually taken as a constant, which in practice ranges from $0.1$ to $10$.  

Any (nontrivial) DP mechanism must be probabilistic by definition.  Thus, the most common measure of the utility of $\mathcal{M}_q(\cdot)$ is its expected $\ell_2$-error:
\begin{align*}
\mathrm{Err}(\mathcal{M}_q,\I) & = \sqrt{\mathsf{E}\left[(\mathcal{M}_q(\I)-|q(\I)|)^2\right]} \\
 &= \sqrt{(\mathsf{E}[\mathcal{M}_q(\I)] - |q(\I)|)^2 + \mathsf{Var}[\mathcal{M}_q(\I)]},
\end{align*}
which consists of a bias term and a variance term.  Between two mechanisms with the same $\mathrm{Err}(\mathcal{M}_q, \I)$, the unbiased one is often preferred.  All our DP mechanisms will be unbiased, in which case $\mathrm{Err}(\mathcal{M}_q, \I) = \sqrt{\mathsf{Var}[\mathcal{M}_q(\I)]}$, while our lower bounds hold even for biased mechanisms.

The problem is now well understood for selection queries with various types of selection conditions \cite{hardt2012simple,barak2007privacy,qardaji2013understanding,qardaji2014practical,xiao2010differential,zhang2014towards,day2016publishing}.
However, the case when $q$ is a conjunctive query (CQ) is still largely open.  Existing solutions \cite{narayan2012djoin,proserpio2014calibrating,johnson2018towards} are heuristics without any notion of optimality.  Notably, there have been extensive studies on graph pattern counting queries \cite{rastogi2009relationship,ding2018privacy,blocki2013differentially,kasiviswanathan2013analyzing,karwa2011private,zhang2015private,chen2013recursive}, which are a special case of CQs equipped with inequalities. But none of them has any theoretical guarantee on the utility, either.  

\subsection{Sensitivity-based DP Mechanisms}
\label{sec:introSen}
Let $\I$ be a database instance and $q$ a counting query. A widely used framework for designing an $\mathcal{M}_q$ is to first compute some measure of \textit{sensitivity} $S_q(\I)$, and then add noise drawn from a certain zero-mean distribution calibrated to $S_q(\I)$.   From now on, we may omit the subscript $q$ from $\mathcal{M}_q$ and $S_q$ when the context is clear.   Note that all sensitivity-based mechanisms are unbiased. Various noise distributions (Section~\ref{sec:preliminary} gives more details) have been studied and they can all achieve $\mathrm{Err(\mathcal{M}, \I)} = \Theta(S(\I))$, where the hidden constant depends on $\varepsilon$ and the particular distribution. Thus, the problem essentially boils down to finding the smallest $S(\I)$ that satisfies the DP requirement. When we say that a particular sensitivity measure $S(\cdot)$ is optimal, this actually means that the DP mechanism $\mathcal{M}(\I)$ that adds noise drawn from a certain distribution calibrated to $S(\I)$ achieves the optimal  (in whatever sense) $\mathrm{Err}(\mathcal{M}, \I)$. 

The main reason for the lack of a good theoretical analysis for CQs under DP is that standard notions of optimality are either useless or unattainable.  It is well known that the \textit{global sensitivity} $GS$, which is defined as the maximum difference between the query answers on \textit{any} two neighboring instances, is worst-case optimal.  Note that $GS$ is an instance-independent sensitivity measure.  However, for any nontrivial CQ involving two or more relations, $GS=\infty$ in strict DP, as changing one tuple may affect an unbounded number of query results.  In relaxed DP, where it is safe to use $N$ in designing the mechanism, $GS$ can be bounded, but can still be as high as $O(N^{n-1})$, where $n$ is the number of relations in $q$. 

For queries where $GS$ is unbounded or too high, a natural idea is to use an instance-specific sensitivity measure \cite{nissim2007smooth,zhang2015private,karwa2011private}.  The \textit{local sensitivity} (definition given in Section \ref{sec:dp_mechanisms}) $LS(\mathbf{I})$ is the most natural choice, and it can be much smaller than $GS$ on most real-world instances.  However, Nissim et al.~\cite{nissim2007smooth} point out that using $LS(\mathbf{I})$ to calibrate noise violates DP; instead, they propose \textit{smooth sensitivity} $SS(\mathbf{I})$, which is always higher than $LS(\mathbf{I})$ but lower than $GS$, and show that it satisfies DP.  

How do we quantify the utility of a DP mechanism that adds instance-specific noise?  Ideally, we would like it to be \textit{instance-optimal} \cite{fagin03:_optim}.  More formally, a DP mechanism $\mathcal{M}(\cdot)$ is said to be \textit{$c$-instance-optimal} if 
\[\mathrm{Err}(\mathcal{M},\I)\leq c \cdot \mathrm{Err}(\mathcal{M}',\I),\]
for every instance $\mathbf{I}$ and every DP mechanism $\mathcal{M}'(\cdot)$.  

Unfortunately, as pointed out by Asi and Duchi \cite{asi2020instance}, instance-optimal DP mechanisms do not exist, unless the query is trivial (i.e., it returns the same count on all instances), because for an $\mathbf{I}$, one can always design a trivial DP mechanism $\mathcal{M}'(\cdot) \equiv |q(\I)|$.  It works perfectly on $\I$ but fails miserably on other instances.  Yet, such a trivial $\mathcal{M}'$ rules out instance-optimal mechanisms. 

\subsection{Neighborhood-optimal DP Mechanisms}

To eliminate such a trivial $\mathcal{M}'$, Asi and Duchi \cite{asi2020instance} propose the following natural relaxation of instance-optimality, which requires $\mathcal{M}'$ to not just work well for $\I$, but also in its neighborhood.

\begin{definition}[$(r,c)$-neighborhood optimal DP mechanisms]
\label{def:optimal_dp_mech}
An $\varepsilon$-DP mechanism $\mathcal{M}(\cdot)$ is said to be \textit{$(r,c)$-neighborhood optimal} if for any instance $\I$ and any $\varepsilon$-DP mechanism $\mathcal{M}'(\cdot)$, there exists an instance $\I'$ with $d(\I,\I')\leq r$ such that
\begin{equation}
\label{eq:neighborOPTdef}
\mathrm{Err}(\mathcal{M},\I)\leq c\cdot \mathrm{Err}(\mathcal{M}',\I').
\end{equation}
\end{definition}

Note that neighborhood optimality smoothly interpolates between instance optimality ($r=0$) and worst-case optimality ($r=N$). It is also called \textit{local minimax optimality} in \cite{asi2020instance}, as it minimizes (up to a factor of $c$) the maximum error in the local neighborhood of $\I$.  

We adopt the convention of \textit{data complexity}, i.e., all asymptotic notations suppress dependencies on the query size.  We also take $\varepsilon$ to be a constant, as with most work on differential privacy.  Thus, when we say that an $\varepsilon$-DP mechanism is $(O(1), O(1))$-neighborhood optimal, this means that $r$ and $c$ may depend on $q$ and $\varepsilon$, but are independent of $N$.  As we are most interested in constant-factor approximations, we often use ``$r$-neighborhood optimal'' as a shorthand of ``$(r, O(1))$-neighborhood optimal''.

While more relaxed than instance optimality, neighborhood optimality is still not easy to achieve (for small $r$).  For example, we can show that the {\sc Median} query (i.e., returning the median of $N$ elements in $[0, 1]$, for odd $N$) does not have any $(r,c)$-neighborhood optimal mechanisms for $r\le\lfloor N/2 \rfloor$ and any $c$. To see this, consider $\I=(0,\dots, 0)$ and $\mathcal{M}'(\cdot)\equiv 0$.  Note that all instances in the $r$-neighborhood of $\I$ have output $0$, so $\mathrm{Err}(\mathcal{M}',\I')=0$ for all $\I'$ with $d(\I, \I')\le r$.  Thus, we must set $\mathcal{M}(\I)=0$ to satisfy \eqref{eq:neighborOPTdef}.  We can apply the same argument on $\I''=(1, \dots, 1)$ and conclude that $\mathcal{M}(\I'')=1$. We have thus found two instances on which $\mathcal{M}(\cdot)$ returns different deterministic values, thus $\mathcal{M}$ cannot be DP by a standard argument \cite{dwork2014algorithmic}.

Acute readers would realize that the negative result for {\sc Median} relies on the fact that  there are certain instances (like $(0,\dots,0)$) with a ``flat'' neighborhood, i.e., the query output is the same within the neighborhood. Fortunately, for full CQs, which are the focus of this paper, these flat neighborhoods do not exist.  This is because in any instance $\I$, one can always add/remove a constant number (depending on $q$) of tuples to change $|q(\I)|$. However, this is merely a necessary condition for a query to admit neighborhood-optimal DP mechanisms.  To actually design such a mechanism $\mathcal{M}(\cdot)$, we need an upper bound on $\mathrm{Err}(\mathcal{M},\I)$, as well as a neighborhood lower bound on  $\min_{\mathcal{M}'}\max_{\I': d(\I,\I')\le r} \mathrm{Err}(\mathcal{M}',\I')$.  Note that both the upper bound and the lower bound are instance specific (i.e., functions of $\I$), and the worst-case (over all $\I$) gap between the upper and lower bounds would become the optimality ratio $c$. 

\subsection{DP Mechanisms for Full CQs}

We relate both the upper and the lower bounds to the smooth sensitivity $SS(\cdot)$. On the lower bound side, in Section \ref{sec:optimality_analysis}, we show that $SS(\cdot)$ is an $O(1)$-neighborhood lower bound.  Thus, the smooth sensitivity based mechanism \cite{nissim2007smooth} is $O(1)$-neighborhood optimal. However, $SS(\cdot)$ in general requires exponential time to compute.  In Section \ref{sec:rs}, we extend the \textit{residual sensitivity} $RS(\cdot)$ \cite{dong21:residual} to arbitrary full CQs, and show that it (1) satisfies DP, (2) is a constant-factor upper bound of $SS(\cdot)$, and (3) can be computed in polynomial time. Together with the lower bound, this yields our first main result:

\begin{theorem}
\label{th:fullCQ}
For any full CQ $q$, there is an $\varepsilon$-DP mechanism $\mathcal{M}(\cdot)$ that is $O(1)$-neighborhood optimal. For any instance $\I$, $\mathcal{M}(\I)$ can be computed in $\mathrm{poly}(N)$ time. 
\end{theorem}

\subsection{CQs with Predicates}
\label{sec:introcqf}

Next, we consider CQs with predicates.  A \textit{predicate} $P(\mathbf{y})$ is a computable function $P:\mathbf{dom}(\mathbf{y}) \rightarrow\{\mathsf{ True, False}\}$ for a set of variables $\mathbf{y}$.  A valuation of $\mathbf{y}$ is a valid query result only if $P(\mathbf{y})$ evaluates to $\mathsf{True}$. 

Predicates are important for expressing graph pattern counting queries. Suppose we would like to count the number of length-3 paths (no repeated vertices) in a directed graph whose edges are stored in a relation $\mathsf{Edge}$.
However, the CQ $\mathsf{Edge}(x_1,x_2)\Join \mathsf{Edge}(x_2,x_3) \Join \mathsf{Edge}(x_3, x_4)$ would not just count the number of length-3 paths, but also length-1 paths, length-2 paths, and triangles.  We have to equip the CQ with inequalities, i.e., $x_i\ne x_j$ for all $i\ne j$, to exclude these false positives. Another type of predicates are comparisons, i.e., $x_i \le x_j$ or $x_i< x_j$, which are common in queries over spatiotemporal databases. 

In Section~\ref{sec:rs_filter} we show how to modify $RS(\cdot)$ to take these predicates into consideration, while preserving its neighborhood optimality. The idea is conceptually easy: we just materialize each predicate $P(\mathbf{y})$ into a relation $\{ t \in \dom(\mathbf{y})\mid P(t)\}$, and then apply our DP mechanism in Theorem~\ref{th:fullCQ}. However, this poses a computational issue, since this relation can be infinite (assuming infinite domains).  To address this issue, we show that it is actually possible to compute $RS(\cdot)$ without materializing the $P(\mathbf{y})$'s.

\begin{theorem}
\label{th:filters}
For any full CQ $q$  with predicates $P_1(\mathbf{y}_1), \dots,  P_\kappa(\mathbf{y}_\kappa)$, there is an $O(1)$-neighborhood optimal $\varepsilon$-DP mechanism $\mathcal{M}(\cdot)$, if for any $\mathbf{z} \subseteq var(q)$, the satisfiability of $\varphi_1 \wedge \cdots \wedge\varphi_S$ is decidable, where each $\varphi_i$ is $P_j(\mathbf{u}_i)$ for any $j$ and $\mathbf{u}_i$ is $\mathbf{y}_j$ after replacing all variables not in $\mathbf{z}$ by any constants.  Furthermore, $\mathcal{M}(\I)$ can be computed in $\mathrm{poly}(N)$ time if all predicates are inequalities or comparisons. 
\end{theorem}

\subsection{Non-full CQs}

To complete the picture, we finally study non-full CQs in Section~\ref{sec:non_full_cq}.  Prior work on non-full CQs simply ignores the projection, and uses the sensitivity of the   full CQ to calibrate noise.  We show how to extend $RS(\cdot)$ to handle the projection more effectively so as to reduce the noise.  Unfortunately, our lower bound no longer holds for non-full CQs, hence losing neighborhood optimality. However, we show that this is unavoidable.  In particular, even for the simple non-full query $\pi_{x_2}(R_1(x_1)\Join R_2(x_1, x_2))$, we show that it does not admit any $o(\sqrt{N})$-neighborhood optimal DP mechanisms. 

\subsection{Related Work}
\label{sec:related_work}

The notion of neighborhood optimality has been recently proposed in the machine learning community \cite{asi2020instance}.  However, their focus is on high-dimensional queries whose output domain must be continuous.  In particular, their lower bound does not hold for any query that returns discrete outputs, in particular, counting queries.  Our lower bound (Section~\ref{sec:general_lower_bound}) holds for arbitrary 1-dimensional queries (with either continuous or discrete outputs), and is also better than their lower bound.  On the upper bound side, their algorithm in general requires exponential time to compute, although they show how it can be made to run in polynomial time for some machine learning problems like mean-estimation and PCA.  

Our work builds upon smooth sensitivity $SS(\cdot)$ \cite{nissim2007smooth} and residual sensitivity $RS(\cdot)$ \cite{dong21:residual}.  Although $SS(\cdot)$ has been shown to preserve DP in \cite{nissim2007smooth}, two issues remain: whether it is optimal (in any sense) and whether it can be computed in polynomial time.  We provide positive answers to both questions in the context of CQs.  While our optimality of $SS(\cdot)$ is the first of its kind, the computational issue has been addressed for some specific queries, such as {\sc Median} (but $SS(\cdot)$ is not neighborhood-optimal for {\sc Median}!), triangle counting \cite{nissim2007smooth}, and $t$-star counting \cite{karwa2011private}, the latter two being special cases of CQs (with inequalities).  However, it is still an open problem whether $SS(\cdot)$ can be computed in polynomial time for an arbitrary CQ in terms of data complexity.  To dodge the computational difficulty, elastic sensitivity $ES(\cdot)$ \cite{johnson2018towards} was introduced as a replacement. However, in Section~\ref{sec:es} we show that $ES(\cdot)$ is not even worst-case optimal. Residual sensitivity $RS(\cdot)$ was proposed in \cite{dong21:residual}, and it has been shown to be a constant-factor approximation of $SS(\cdot)$ while being polynomially computable. However, $RS(\cdot)$ as defined in \cite{dong21:residual} does not allow self-joins, which are challenging under DP, since changing one tuple corresponds to changing multiple logical relations when there are self-joins.  This can be tricky for both the privacy guarantee and the mechanism's optimality.  
Second, no lower bound was provided in \cite{dong21:residual}, so the optimality of $RS(\cdot)$ is not known until this paper.  Thirdly, we extend $RS(\cdot)$ to CQs with predicates and non-full CQs.

The DP policy considered in this paper treats a change as the insertion, deletion, or substitution of a tuple, hence also called \textit{tuple-DP}.  In \textit{user-DP}~\cite{kotsogiannis2019privatesql,tao2020computing}, there is a designated \textit{primary private relation} $R_P \in \mathbf{R}$ that contains all the users, while tuples in other relations with a foreign key (FK) referencing the primary key (PK) of a tuple (user) $t_P \in \I(R_P)$, directly or indirectly, are considered as data belonging to $t_P$.  Two instances $\I$ and $\I'$ are considered neighbors if $\I'$ can be obtained from $\I$ by adding/deleting/substituting a set of tuples, all of which reference the same user $t_P \in \I(R_P)$. Applying user-DP and tuple-DP on the graph schema $\mathbf{R} = \{\mathtt{Node(ID)}$, $\mathtt{Edge(src,dst)}\}$, where $\mathtt{src}$ and $\mathtt{dst}$ are both FKs referencing $\mathtt{ID}$, yields the two well-known DP policies for graph data: node-DP \cite{kasiviswanathan2013analyzing} and edge-DP  \cite{karwa2011private}.  By designating $\mathtt{Node}$ as the primary private relation, user-DP on $\mathbf{R}$ becomes node-DP; by setting $\mathbf{R}_P = \{\mathtt{Edge}\}$, tuple-DP becomes edge-DP (note that in tuple-DP, we ignore the FK constraints, which effectively means that $\mathtt{Node}$ is irrelevant). User-DP is more general than tuple-DP, hence potentially incurring a higher privacy cost.
Very recently, a logarithmic-neighborhood optimal mechanism has been proposed for CQs under user-DP \cite{dong22:R2T}.  Meanwhile, it has also been shown that $O(1)$-neighborhood optimality is not achievable under user-DP \cite{dong2021universal} even for the simple query $R_1(x_1) \Join R_2(x_1, x_2)$ where $R_1$ is the primary private relation and $R_2.x_1$ is an FK referencing $R_1.x_1$.  Therefore, combined with the result obtained in this paper, we now have a separation for CQs under tuple-DP and user-DP.

\section{Preliminaries}
\label{sec:preliminary}

\subsection{Conjunctive Queries}
\label{sec:conjunctive_query}

Let $\mathbf{R}$ be a database schema. A \textit{full conjunctive query (CQ)} has the form
\[ q :=  R_1(\x_1) \Join \cdots \Join R_n(\x_n),
\]
where $R_1,\dots,R_n$ are relation names in $\mathbf{R}$, and each $\mathbf{x}_i$ is a set of $arity(R_i)$ variables/attributes\footnote{If $\mathbf{x}_i$ has constants, we can preprocess $R_i(\mathbf{x}_i)$ in linear time so that only tuples that match these constants remain. }. We call each $R_i(\mathbf{x}_i)$ an \textit{atom}.  We use $[n]$ to denote $\{1,\dots,n\}$, and $[i,j] =\{i,\dots,j\}$. For any $E\subseteq[n]$, $\Bar{E}=[n]-E$.  For a variable $x$, we use $\dom(x)$ to denote the domain of $x$.  For $\mathbf{x} = (x_1, \dots, x_k)$, let $\dom(\mathbf{x}) = \dom(x_1) \times \cdots \times \dom(x_k)$.  Let $var(q)$ denote the set of variables in $q$.

When considering self-joins, there can be repeats, i.e., $R_i=R_j$. In this case, we assume $\mathbf{x}_i \ne \mathbf{x}_j$; otherwise one of the two atoms is redundant. 
Let $\mathbf{I}$ be a database instance over $\mathbf{R}$.  For a relation name $R\in \mathbf{R}$, let $\mathbf{I}(R)$ denote the relation instance of $R$.  We use $I_i$ as a shorthand for $\mathbf{I}(R_i)$.  $\I$ and the $I_i$'s are called \textit{physical instances}.  On the other hand, we use $I_i(\mathbf{x}_i)$ to denote $I_i$ after renaming its attributes to $\mathbf{x}_i$.  The $I_i(\mathbf{x}_i)$'s are called the \textit{logical instances}.  Note that if $R_i$ and $R_j$ are the same relation name, then $I_i=I_j$, but $I_i(\mathbf{x}_i) \ne I_j(\mathbf{x}_j)$, as $I_i(\mathbf{x}_i)$ and $I_j(\mathbf{x}_j)$ have different attributes. For a self-join-free query, we may without loss of generality assume that $\mathbf{x}_i=sort(R_i)$ for all $i\in [n]$ so the logical instances are the same as the physical instances, but for queries with self-joins, one physical relation instance may correspond to multiple logical relation instances.  This distinction makes the problem more difficult under DP, as the distance between two logical instances may be larger than between the physical instances.

By rearranging the atoms, we may assume that all appearances of the same relation name are consecutive. 
Suppose $m$ distinct relation names are mentioned in $q$, and for $i=1,\dots,m$, $R_{l_i}, \dots, R_{l_{i+1}-1}$ are the same relation name (set $l_{m+1} = n+1$). Let $D_i=[l_i,l_{i+1}-1]$ and $n_i=l_{i+1}-l_{i}$, which is the number of copies of $R_{l_i}$ mentioned in $q$.

\subsection{Differential Privacy in Relational Databases}
\label{sec:dp}

Differential privacy is already defined in \eqref{eq:DPdef}.  This notion can be applied to any problem by properly defining the distance function $d(\cdot,\cdot)$. As a database may contain both public and private relations, we use a more refined definition of $d(\cdot, \cdot)$.  For two relation instances over the same relation name $I, I'$, we use $d(I,I')$ to denote the distance between $I$ and $I'$, which is the minimum number of steps to change $I$ into $I'$.
Note that $d(I_j,I_j')$ is the same for all $j\in D_i$, for any $i \in [m]$, as they are the same physical relation instance. 

We use $P^m\subseteq[m]$ to denote the set of private physical relations, while $P^n=\cup_{i\in P^m}D_i$ is the set of private logical relations.  Let $m_P=|P^m|$, $n_P=|P^n|$. 
Two database instances can only differ in the private relations, i.e., $d(I_j, I'_j)=0$ for every $j\in \bar{P}^n$. 
In the DP definition \eqref{eq:DPdef}, we must use the distance between the physical database instances, i.e., $d(\mathbf{I},\mathbf{I}') = \sum_{i\in [m]} d(I_{l_i},I_{l_i}')$.  Note that the distance between the logical instances, namely $\sum_{i\in [n]} d(I_{i}(\mathbf{x}_i),I_{i}'(\mathbf{x}_i))$, can be larger than $d(\mathbf{I},\mathbf{I}')$ when self-joins are present. 

A simple but important observation is that a query with self-joins on instance $\{I_i\}_i$ can be considered as a query without self-joins on instance $\{I_i(\mathbf{x}_i)\}_i$.  This allows us to reuse some of the technical results from \cite{dong21:residual} on self-join-free queries.  However, the critical difference is that this conversion enlarges the distance, while the DP guarantee must hold with respect to the distance on the original, physical instance.

\subsection{Sensitivity-based DP Mechanisms}
\label{sec:dp_mechanisms}

As mentioned in Section~\ref{sec:intro}, the most common technique for designing DP mechanisms is to first compute some measure of sensitivity $S(\cdot)$ of the query, and then add noise drawn from a certain distribution calibrated to $S(\cdot)$. 

The \textit{local sensitivity} of $q$ at instance $\I$ is how much $|q(\I)|$ can change if we change one tuple in $\I$, i.e.,
\begin{equation}
\label{eq:ls}
LS(\mathbf{I})=\max_{\mathbf{I}', d(\mathbf{I}, \mathbf{I}')=1}\left| |q(\mathbf{I})|-|q(\mathbf{I}')| \right|.
\end{equation}

However, using $LS(\cdot)$ directly breaches privacy \cite{nissim2007smooth}, because two neighboring instances may have very different $LS(\cdot)$.  On the other hand, it is safe to use the \textit{global sensitivity}, which is the maximum $LS(\cdot)$ over all instances:
\begin{equation*}
\label{eq:gs}
GS=\max_{\mathbf{I}}LS(\mathbf{I}).
\end{equation*}
It is well known that one can achieve $\varepsilon$-DP with $\mathrm{Err}(\mathcal{M}, \I) =O(GS)$ by drawing noise from the Laplace distribution calibrated to $GS/\varepsilon$~\cite{dwork2006calibrating}. However, this loses the utility: the $GS$ for CQs can be as large as $O(N^{n_P-1})$ under relaxed DP, and $\infty$ under strict DP.

To address this issue, Nissim et al. \cite{nissim2007smooth} propose the \textit{smooth sensitivity} $SS_\beta(\cdot)$. To define $SS_\beta(\cdot)$, we first define the \textit{local sensitivity at distance $k$}:
\begin{equation}
\label{eq:lsk}
LS^{(k)}(\mathbf{I})=\max_{\mathbf{I}', d(\mathbf{I}, \mathbf{I}')\leq k}LS(\mathbf{I}').
\end{equation}
Note that for CQs, (\ref{eq:lsk}) can be rewritten as
\begin{equation}
\label{eq:lsk_cq}
LS^{(k)}(\mathbf{I})=\max_{\mathbf{I}', d(\mathbf{I}, \mathbf{I}')= k}LS(\mathbf{I}'),
\end{equation}
because if $d(\mathbf{I}, \mathbf{I}')<k$, we can always insert dummy tuples so that $d(\mathbf{I}, \mathbf{I}')=k$.
Then for a parameter $\beta$, the smooth sensitivity is \begin{equation}
\label{eq:ss}
SS_\beta(\mathbf{I})=\max_{k\geq0} e^{-\beta k} LS^{(k)}(\mathbf{I}).
\end{equation}

It is clear that $SS_\beta(\I)\le GS$.  More importantly, $SS_\beta(\cdot)$ is ``smooth'', in the sense that $SS_\beta(\I)\leq e^{\beta} SS_\beta(\I')$ for any two neighboring instances $\I$ and $\I'$. Due to its smoothness, it has been shown that setting $\beta= \varepsilon/10$ and drawing noise calibrated to $SS_\beta(\I)/\beta$ from a general Cauchy distribution, which has pdf $h(z)\propto \frac{1}{1+|z|^4}$, achieves $\varepsilon$-DP with $\mathrm{Err}(\mathcal{M}, \I) = \frac{SS_\beta(\I)}{\beta} = O(SS_\beta(\I))$.  The choice of the constant $10$ is arbitrary, which affects the tail properties of the noise distribution, but not the variance (asymptotically).  In the rest of the paper, we omit the subscript $\beta$ from $SS_\beta(\cdot)$ for brevity.

However, computing $SS(\I)$ is very costly. The definition \eqref{eq:ss} does not yield an efficient (or even computable) algorithm. To address this issue, Nissim et al. \cite{nissim2007smooth} show that any smooth upper bound of $SS(\cdot)$ can be used.  Specifically, let
\begin{equation}
\label{eq:hat_ss}
\hat{SS}(\I) = \max_{k\geq 0}e^{-\beta k}\hat{LS}^{(k)}(\I).
\end{equation}
It has been shown that as long as $\hat{LS}^{(k)}(\cdot)$ is an upper bound of ${LS}^{(k)}(\cdot)$ and satisfies the smoothness property, i.e., for any neighbors $\I$ and $\I'$,
\begin{equation}
\label{eq:des_property}
\hat{LS}^{(k)}(\I)\leq \hat{LS}^{(k+1)}(\I'),
\end{equation}
then one can use $\hat{SS}(\cdot)$ in place of $SS(\cdot)$ to calibrate noise while preserving $\varepsilon$-DP.  The error becomes $\mathrm{Err}(\mathcal{M}, \I) = \frac{10 \hat{SS}(\I)}{\varepsilon}$ accordingly.

\section{Residual Sensitivity for Full CQ\lowercase{s}}
\label{sec:rs}

\subsection{Residual Queries}
\label{sec:rs_self_join_free}

The residual sensitivity $RS(\cdot)$ can be considered as an instantiation of $\hat{SS}(\cdot)$ for CQs, i.e., it is defined as in \eqref{eq:hat_ss} with a particular choice of $\hat{LS}^{(k)}(\cdot)$ that has the smoothness property (\ref{eq:des_property}). Our $\hat{LS}^{(k)}(\cdot)$ is based on the \textit{residual queries} of a given CQ $q$.

A \textit{residual query} of $q$ on a subset $E\subseteq[n]$ of relations is $q_E := \,\Join_{i\in  E} R_i(\mathbf{x}_i)$.  Its \textit{boundary}, denoted $\partial q_E$, is the set of attributes that belong to atoms both in and out of $E$, i.e., $\partial q_E = \{ x \mid x\in \mathbf{x}_i \cap \mathbf{x}_j,  i \in E, j \in \bar{E} \}$.  For a \textit{residual} query $q_E$ on database instance $\I$, its \textit{maximum multiplicity over the boundary} is defined as
\begin{equation*}
\label{eq:TE}
T_{E}(\mathbf{I}) = \max_{t \in \dom(\partial q_E)}  |q_E(\mathbf{I}) \Join t |.
\end{equation*}
A \textit{witness tuple} of the maximum multiplicity over the boundary of $q_E(\I)$ is 
\begin{equation}
\label{eq:te}
t_E(\mathbf{I})=\argmax_{t \in \dom(\partial q_E)}|q_E(\mathbf{I}) \Join t |.
\end{equation}
Per convention, when $E=\emptyset$, we set $q_E(\mathbf{I})=\{\langle \rangle\}$, where $\langle \rangle$ denotes the empty tuple, so $T_{\emptyset}(\I)=1$. 

We note that $T_E(\I)$ is exactly an AJAR/FAQ query \cite{joglekar2016ajar,abo2016faq}, where the annotations of all tuples are $1$ with two semiring aggregations $+$ and $\max$.  The $+$ aggregation is done group-by $\partial q$, followed by a $\max$ over all the $+$ aggregates. Such a query can be computed in $O(N^w)$ time, where $w$ is its AJAR/FAQ width, a constant depending on the query only. 

To develop $RS(\cdot)$, we need some properties of $T_E(\cdot)$. We start with a simple one:

\begin{lemma}
\label{lm:simple_lemma}
Given any CQ $q$, $E\subseteq[n]$, and two database instances $\I,\I'$, if $I_i(\x_i) \subseteq I'_i(\x_i)$ for all $i\in E$, then $T_E(\mathbf{I})\le T_E(\mathbf{I}')$. In particular, if $I_i(\x_i) = I'_i(\x_i)$ for all $i\in E$, then $T_E(\mathbf{I})= T_E(\mathbf{I}')$.
\end{lemma}

As will become clear, the sensitivity of $q$ depends on the sensitivity of $T_E(\cdot)$.  The following technical lemma provides such an upper bound:

\begin{lemma}
\label{lm:sensitivity_TE}
For any CQ $q$, $E\subseteq[n]$, and two instances $\mathbf{I}, \mathbf{I}'$,
\begin{equation*}
|T_E(\mathbf{I})-T_E(\mathbf{I}')|\leq \sum_{E'\subseteq E,E'\neq \emptyset} \left( T_{E-E'}(\mathbf{I}) \prod_{i\in E'} d(I_i(\x_i), I'_i(\x_i))\right).
\end{equation*}
\end{lemma}

\begin{proof}
The self-join-free version of this lemma was proved in \cite{dong21:residual}.  Recall that a query with self-joins can be considered as a self-join-free query on the logical instance. This means that $T_E(\{I_i\}_i) =T_E(\{I_i(\mathbf{x}_i)\}_i) $.   We now have
\begin{align*}
|T_E(\mathbf{I})-T_E(\mathbf{I}')|=&|{T}_E(\{I_i(\mathbf{x}_i)\}_i)-{T}_E(\{I'_i(\mathbf{x}_i)\}_i)|
\\
\leq & \sum_{E'\subseteq E,E'\neq \emptyset} \left( T_{E-E'}(\{I_i(\mathbf{x}_i)\}_i) \prod_{i\in E'} d({I}_i(\x_i), {I}'_i(\x_i))\right)
\\
=&\sum_{E'\subseteq E,E'\neq \emptyset} \left( T_{E-E'}(\mathbf{I}) \prod_{i\in E'} d(I_i(\x_i), I'_i(\x_i))\right),
\end{align*}
where the inequality invokes the self-join-free version of the lemma, because on $\{I_i(\mathbf{x}_i)\}_i$, the query is self-join-free.
\end{proof}

\subsection{Local Sensitivity of CQs}
\label{sec:ls}

We first consider the local sensitivity $LS(\cdot)$.  For a CQ without self-joins, its local sensitivity is precisely characterized by the $T_E(\cdot)$'s. 

\newcommand{\overbar}[1]{\mkern 1.5mu\overline{\mkern-5mu#1\mkern-5mu}\mkern 1.5mu}

\begin{lemma}[\cite{dong21:residual}]
\label{lm:ls_self_join_free_query}
For a CQ without self-joins, 
\[LS(\I) =  \max_{i\in P^n}T_{\overbar{\{i\}}}(\I).
\]
\end{lemma}

To extend this result to CQs with self-joins, we need to bound how much $|q(\I)|$ can change when multiple relations change simultaneously, as a change in one physical relation instance may correspond to changes in multiple logical relations when self-joins are present.  We first consider the case for self-join-free queries.

\begin{lemma}
\label{lm:base_ls_self_join_query}
Let $q$ be a CQ $q$ without self-joins, $B\subseteq [n]$, $B\neq \emptyset$, and let $\mathbf{I},\mathbf{I}'$ be instances such that $d(I_j,I_j')=1$ for all $j\in B$ while $d(I_j,I_j')=0$ otherwise.  Then 
\[\left||q(\mathbf{I})|-|q(\mathbf{I'})|\right| \leq \sum_{E\subseteq B,E\neq \emptyset}T_{\bar{E}}(\mathbf{I}).\] 
\end{lemma}

\begin{proof}
For the given $\I,\I'$, there is a sequence of intermediate instances $\mathbf{I}^0, \mathbf{I}^1,\dots,\mathbf{I}^n$ such that, $\mathbf{I}^0=\I$, $\mathbf{I}^n=\I'$, while $\I^{\ell-1},\I^{\ell}$  can only possibly differ in $I_{\ell}$, $\ell=1,2,\dots,n$. More precisely, we set $I_{j}^{\ell} = I_{j}'$ if $j\leq \ell$; otherwise $I_{j}^{\ell} = I_{j}$.

Notice that, for all $j\notin B,\mathbf{I}^{j}=\mathbf{I}^{j-1}$, that is 
\[\left||q(\mathbf{I}^{j})|-|q(\I^{j-1})|\right|=0.\]
For all $j\in B$, $d(\mathbf{I}^{j},\mathbf{I}^{j-1})=d(I^{j}_j,I^{j-1}_j)=1$. By Lemma~\ref{lm:ls_self_join_free_query},
\begin{equation*}
\left||q(\mathbf{I}^{j})|-|q(\mathbf{I}^{j-1})|\right|\leq T_{[n]-\{j\}}(\mathbf{I}^{j-1}).
\end{equation*}
Therefore,
\begin{equation}
\label{eq:lm:base_ls_self_join_query_1}
\left||q(\mathbf{I})|-|q(\mathbf{I'})|\right|=\left||q(\mathbf{I}^{n})|-|q(\mathbf{I^{0}})|\right|\leq \sum_{j\in B}T_{[n]-\{j\}}(\mathbf{I}^{j-1}).
\end{equation}

Note that for a self-join-free query, $I_i = I_i(\mathbf{x}_i)$.  By Lemma~\ref{lm:sensitivity_TE}, for all $j\in B$,
\begin{align}
&T_{[n]-\{j\}}(\mathbf{I}^{j-1})
\nonumber
\\
\nonumber
\leq& T_{[n]-\{j\}}(\mathbf{I}^{0})+\sum_{E\subseteq [n]-\{j\},E\neq \emptyset}\left(T_{[n]-\{j\}-E}(\mathbf{I}^0)\prod_{t\in E}d(I_t^{j-1},I_j^0)\right)
\\
= &\sum_{E\subseteq [n]-\{j\}}\left(T_{[n]-\{j\}-E}(\mathbf{I}^0)\prod_{t\in E}d(I_t^{j-1},I_j^0)\right)
\nonumber
\\
=&\sum_{E\subseteq B\cap [j-1]}T_{[n]-\{j\}-E}(\mathbf{I}^0)
\label{eq:lm:base_ls_self_join_query_2}
\\
=& \sum_{E\subseteq B\cap [j],j\in E}T_{\bar{E}}(\mathbf{I}^0).
\label{eq:lm:base_ls_self_join_query_3}
\end{align}
(\ref{eq:lm:base_ls_self_join_query_2}) is because, by definition, for $\mathbf{I}^{j-1},\mathbf{I}^{0},j\in B$, $d(I_t^{j-1},I_t^0)=1$ if $t\in B\cap [j-1]$; otherwise $d(I_t^{j-1},I_t^0)=0$. 

Plugging (\ref{eq:lm:base_ls_self_join_query_3}) into (\ref{eq:lm:base_ls_self_join_query_1}), we obtain
\[\left||q(\mathbf{I})|-|q(\mathbf{I'})|\right|\leq\sum_{j\in B}\sum_{E\subseteq B\cap[j],j\in E}T_{\bar{E}}(\mathbf{I}^0)= \sum_{E\subseteq B,E\neq \emptyset}T_{\bar{E}}(\mathbf{I}).\]\end{proof}

Based on this and (\ref{eq:ls}), we can derive an upper bound on $LS(\I)$ for CQs with self-joins.

\begin{theorem}
\label{th:ls_self_join_query}
For a CQ $q$, 
\[LS(\I)\leq \max_{i\in P^m}\sum_{E\subseteq D_i,E\neq \emptyset}T_{\bar{E}}(\I).\]
\end{theorem}

\begin{proof}
First, we write (\ref{eq:ls}) as 
\begin{equation}
LS(\I) = \max_{i\in P^m}\max_{\I',d(\I,\I')=1,d(I_{l_i},I_{l_i}')=1}\left||q(\I)|-|q(\I')|\right|.
\label{eq:th:ls_self_join_query_1}
\end{equation}

We can use a similar idea in the proof of Lemma~\ref{lm:sensitivity_TE}: regard the $q$ as a self-join-free query on the logical instance and any $\I,\I'$ with $d(\I,\I')=1,d(I_{l_i},I_{l_i}')=1$ correspond to $\{I_j(\x_j)\}_j,\{I_j'(\x_j)\}_i$ with $d(I_j(\x_j),I_j(\x_j))=1$ for all $j\in D_i$. 

Then, we can derive, for any $\I,\I'$ with $d(\I,\I')=1,d(I_{l_i},I_{l_i}')=1$,
\begin{align}
\left||q(\I)|-|q(\I')|\right|=& \left||q(\{I_j(\x_j)\}_j)|-|q(\{I_j'(\x_j)\}_j)|\right|
\nonumber
\\
\label{eq:th:ls_self_join_query_2}
\leq &\sum_{E\subseteq D_i,E\neq \emptyset}T_{\bar{E}}(\{I_j(\x_j)\}_j)
\\
\label{eq:th:ls_self_join_query_3}
=& \sum_{E\subseteq D_i,E\neq \emptyset}T_{\bar{E}}(\I),
\end{align}
where (\ref{eq:th:ls_self_join_query_2}) follows from Lemma~\ref{lm:base_ls_self_join_query}.

Finally, plugging (\ref{eq:th:ls_self_join_query_3}) into (\ref{eq:th:ls_self_join_query_1}), we obtain
\[LS(\I)\leq \max_{i\in P^m} \sum_{E\subseteq D_i,E\neq \emptyset}{T}_{\bar{E}}({\mathbf{I}}).\]\end{proof}

\paragraph{Remark}
Note that when self-joins are present, we can no longer obtain an exact formula for $LS(\I)$ like for self-join-free queries in Lemma~\ref{lm:ls_self_join_free_query}.  This is precisely due to the fact that self-joins induce changes in multiple logical relations that may interact in complex manners.

\subsection{Global Sensitivity of CQs}
\label{sec:gs_cq}
Because $GS=\max_{\I} LS(\I)$, a by-product of Theorem~\ref{th:ls_self_join_query} is an upper bound on $GS$ under relaxed DP where the instance size $N$ is public.  This upper bound can be much smaller than the trivial upper bound $O(N^{n_P-1})$ mentioned in Section \ref{sec:dp_mechanisms} for many CQs.

By Theorem~\ref{th:ls_self_join_query}, we have
\begin{equation}
GS   \leq \max_{\mathbf{I}}\max_{i\in P^m}\sum_{E\subseteq D_i,E\neq \emptyset}T_{\bar{E}}(\I) \leq \max_{i\in P^m}\sum_{E\subseteq D_i,E\neq \emptyset}\max_{\mathbf{I}}T_{\bar{E}}(\I).
\label{eq:hat_gs}
\end{equation}
Observe that $\max_{\mathbf{I}}T_{\bar{E}}(\I)$ is upper bounded by the maximum join size of the residual query $q_{\bar E}$, when the logical relations of the same physical relation are allowed to be instantiated differently and the domain size of each variable in $\partial q_E = \partial q_{\bar E}$ is set to $1$, which is equivalent to removing these variables.  We can bound the maximum join size using the \textit{AGM bound}~\cite{atserias2008size}. Together with \eqref{eq:hat_gs} this yields an upper bound on $GS$.

\begin{example}
We illustrate how this is done on the triangle query $q=\mathtt{Edge}(x_1,x_2)\Join \mathtt{Edge}(x_2,x_3)\Join \mathtt{Edge}(x_1,x_3)$ on a single   physical relation $\mathtt{Edge}$. For $E =\{3\}$, i.e., $q_{\bar E}=\mathtt{Edge}(x_1,x_2)\Join \mathtt{Edge}(x_2,x_3)$ and $\partial q_{E} = \{x_1,x_3\}$, we have
\begin{align*}
\max_{\mathbf{I}}T_{\bar E}(\I)=&\max_{\mathbf{I}}\max_{t\in \dom(x_1,x_3)}|\mathtt{Edge}(x_1,x_2)\Join \mathtt{Edge}(x_2,x_3)|
\\
\leq &\max_{\mathbf{I}}\max_{t\in \dom(x_1,x_3)}|\mathtt{Edge}_1(x_1,x_2)\Join \mathtt{Edge}_2(x_2,x_3)|
\\
=& \max_{\mathbf{I}}(\mathtt{Edge}_1(x_2)\Join \mathtt{Edge}_2(x_2)).
\\
= & \mathsf{AGM}(\mathtt{Edge}_1(x_2)\Join \mathtt{Edge}_2(x_2)).
\end{align*}
We can similarly derive a bound for other $E$'s. Note that when $E$ consists of $2$ relations, $\max_{\mathbf{I}}T_{\bar E}(\I)\le 1$. Thus,
\begin{align*}
GS \le & \mathsf{AGM}(\mathtt{Edge}_1(x_2)\Join \mathtt{Edge}_2(x_2))
\\
&+\mathsf{AGM}(\mathtt{Edge}_2(x_3)\Join \mathtt{Edge}_3(x_3))
\\
&+ \cdots\\
=& O(N).
\end{align*}
\end{example}

\begin{example}
\label{ex:path4}
As a more complicated example, consider the path-4 query
\[q=\mathtt{Edge}(x_1,x_2)\Join \mathtt{Edge}(x_2,x_3)\Join \mathtt{Edge}(x_3,x_4)\Join \mathtt{Edge}(x_4,x_5).\]
We have 
\begin{align*}
GS  \le& \mathsf{AGM}(\mathtt{Edge}_2(x_3) \Join\mathtt{Edge}_3(x_3,x_4) \Join \mathtt{Edge}_4(x_4, x_5)) \\
 &+ \mathsf{AGM}(\mathtt{Edge}_1(x_1) \Join\mathtt{Edge}_3(x_4) \Join \mathtt{Edge}_4(x_4, x_5)) \\
 &+ \mathsf{AGM}(\mathtt{Edge}_1(x_1,x_2) \Join\mathtt{Edge}_2(x_2) \Join \mathtt{Edge}_4(x_5)) \\
 &+ \cdots\\
 =& O(N^2).
\end{align*}
\end{example}

In general, any join size upper bound can be plugged into \eqref{eq:hat_gs}. For example, when degree information or functional dependencies are publicly available, tighter upper bounds can be derived \cite{gottlob12size,abokhamis17shannon}.  Although DP mechanisms based on $GS$ are not as accurate as our $RS(\cdot)$-based mechanisms to be presented next, they can be computed in $O(1)$ time (excluding the time for computing $|q(\I)|$). 

\subsection{Deriving \texorpdfstring{$\hat{LS}^{(k)}(\cdot)$}{TEXT}}
\label{sec:hat_lsk}

Theorem \ref{th:ls_self_join_query} has provided an upper bound for $LS(\I)$.  The next step is to derive $\hat{LS}^{(k)}(\I)$, an upper bound for $LS^{(k)}(\I)$.

For any two instances $\I,\I'$, define their \textit{distance vector} as \[\s=(d(I_{1},I_{1}'),d(I_{2},I_{2}'),\dots,d(I_{n},I_{n}')).\]
For any $\mathbf{s} =(s_1,\dots, s_n)$, let $\mathcal{I}^{\mathbf{s}}=\{\mathbf{I}': d(I_{j},I_{j}')=s_j, j\in[n]\}$ be the set of instances whose distance vectors are $\s$ from $\I$.  Note that when self-joins are present, not any $\s\in \mathbb{N}^n$ is a valid distance vector. We must ensure $s_{l_i} = s_{l_i+1}=\dots=s_{l_{i}+n_i-1}$, for any $i\in[m]$.  Let $\mathbf{S}^k$ be the set of valid distance vectors such that the total distance of all private relations is $k$, i.e., 
\[\mathbf{S}^k=\left\{\mathbf{s}:\sum_{i\in P^m}s_{l_i}=k; s_j=0, j\in \bar{P}^n; \forall i \in [m], j\in D_i, s_{j}=s_{l_i} \right\}.\]
Denote the set of instances at $k$ distance from $\I$ as $\mathcal{I}^k = \{\I':d(\I,\I')=k\}$, i.e., 
\[\mathcal{I}^k = \cup_{\s\in \mathbf{S}^k}\mathcal{I}^{\mathbf{s}}.\]
We now derive an upper bound of $LS^{(k)}(\cdot)$ in terms of $T_E(\cdot)$.
\begin{lemma}
\label{lm:upper_bound_lsk}
\[LS^{(k)}(\I)\leq\max_{\s\in \mathbf{S}^k} \max_{i\in P^m} \sum_{E\subseteq D_i,E\neq \emptyset}\max_{\I'\in \mathcal{I}^{\s}}{T}_{\bar{E}}(\I').\]
\end{lemma}

\begin{proof}
By (\ref{eq:lsk_cq}), we have
\begin{align}
LS^{(k)}(\I) =& \max_{\mathbf{I}', d(\mathbf{I}, \mathbf{I}')= k}LS(\mathbf{I}')
\nonumber
\\
=& \max_{\I'\in \mathcal{I}^k}\max_{i\in P^m} LS(\mathbf{I}')
\nonumber
\\
\leq &\max_{\I'\in \mathcal{I}^k}\max_{i\in P^m}\sum_{E\subseteq D_i,E\neq \emptyset}{T}_{\bar{E}}(\I')
\label{eq:lm:upper_bound_lsk_l1}
\\
=&\max_{\s\in \mathbf{S}^k} \max_{\I'\in \mathcal{I}^{\s}}\max_{i\in P^m}\sum_{E\subseteq D_i,E\neq \emptyset}{T}_{\bar{E}}(\I')
\nonumber
\\
=&\max_{\s\in \mathbf{S}^k} \max_{i\in P^m}\max_{\I'\in \mathcal{I}^{\s}}\sum_{E\subseteq D_i,E\neq \emptyset}{T}_{\bar{E}}(\I')
\nonumber
\\
\leq & \max_{\s\in \mathbf{S}^k} \max_{i\in P^m} \sum_{E\subseteq D_i,E\neq \emptyset}\max_{\I'\in \mathcal{I}^{\s}}{T}_{\bar{E}}(\I').
\end{align}
(\ref{eq:lm:upper_bound_lsk_l1}) follows Theorem~\ref{th:ls_self_join_query}.
\end{proof}

Let $\hat{T}_{E,\mathbf{s}}(\mathbf{I})$ be an upper bound of $\max_{\I'\in \mathcal{I}^{\mathbf{s}}}{T}_{E}(\mathbf{I'})$. Then 
\begin{equation}
\label{eq:hat_lsk}
\hat{LS}^{(k)}(\mathbf{I}):=\max_{\mathbf{s}\in \mathbf{S}^k}\max_{i\in P^m}\sum_{E\subseteq D_i,E\neq \emptyset}\hat{T}_{\bar{E},\s}(\I'),
\end{equation}
is clearly an upper bound of $LS^{(k)}(\I)$. 

Now, it remains to find a valid $\hat{T}_{E,\mathbf{s}}(\mathbf{I})$.
By Lemma~\ref{lm:sensitivity_TE}, we have for any $E\subseteq [n]$ and any $\I'\in \mathcal{I}^{\mathbf{s}}$,
\[T_E(\I')\leq T_{E}(\I) + \sum_{E'\subseteq E, E'\neq \emptyset}\left(T_{E-E'}(\I)\prod_{j\in E'}s_j\right).\]
So we set $\hat{T}_{E,\mathbf{s}}(\mathbf{I})$ as
\begin{equation}
\label{eq:hat_TES}
\hat{T}_{E,\mathbf{s}}(\I):=\sum_{E'\subseteq E}\left(T_{E-E'}(\I)\prod_{j\in E'}s_j\right),
\end{equation}
where we define $\prod_{j\in \emptyset}s_j=1$.

Finally, the residual sensitivity is defined as in \eqref{eq:hat_ss} by setting $\hat{LS}^{(k)}(\I)$ as in \eqref{eq:hat_lsk}:
\begin{equation}
\label{eq:rs}
RS(\I) = \max_{k\ge 0}e^{-\beta k}\hat{LS}^{(k)}(\I).
\end{equation}

However, in order for $RS(\cdot)$ to be a valid DP mechanism, we need to show that $\hat{LS}^{(k)}(\cdot)$ follows the smoothness property (\ref{eq:des_property}). To do so, we need a technical result from \cite{dong21:residual}:

\begin{lemma}[\cite{dong21:residual}]
\label{lm:base_for_proof_desir_pro}
Given any self-join-free CQ, any $E\subseteq[n]$, any $i\in[n]$,  two instances 
$\I,\I'$ that differ by one tuple in $R_i(\x_i)$, any two distance vectors $\s=(s_1,\dots, s_n)$ and $\s'=(s'_1,\dots, s'_n)$ such that
\[\s' = (s_1,\dots,s_{i-1},s_{i}+1,s_{i+1},\dots,s_n),\]
we have
\[\sum_{E'\subseteq E}\left(T_{E-E'}(\I)\prod_{j\in E'}s_j\right)\leq \sum_{E'\subseteq E}\left(T_{E-E'}(\I')\prod_{j\in E'}s_j'\right).\]
\end{lemma}

In order to handle self-joins, we need to extend Lemma~\ref{lm:base_for_proof_desir_pro} to the case where multiple relations may differ.

\begin{lemma}
\label{lm:inter_for_proof_desir_pro}
Given any self-join-free CQ, any $E\subseteq[n]$, any $B\subseteq [n]$, two instances 
$\I,\I'$ that differ by one tuple in every $R_i(\x_i),i\in B$, and two distance vectors  $\s=(s_1,\dots, s_n)$ and $\s'=(s'_1,\dots, s'_n)$ such that
\[s_i' = \begin{cases}
s_i+1, & i\in B;
\\
s_i, & i\notin B,
\end{cases}\]
we have
\[\sum_{E'\subseteq E}\left(T_{E-E'}(\I)\prod_{j\in E'}s_j\right)\leq \sum_{E'\subseteq E}\left(T_{E-E'}(\I')\prod_{j\in E'}s_j'\right).\]
\end{lemma}

\begin{proof}
Given $B\subseteq[n]$, two instances $\I,\I'$ differing by one tuple in all $R_i(\x_i),i\in B$, we define ${\I}^{0},{\I}^{1},\dots,{\I}^{n}$ and ${\s}^{0},{\s}^{1},\dots,{\s}^{n}$: for $\ell\in [0, n]$,
\[{I}^{\ell}_j(\x_j)=\begin{cases}
{I}_j'(\x_j) & j\in[\ell]
\\
{I}_j(\x_j) & otherwise
\end{cases}\]
and
\[{s}^{\ell}_j=\begin{cases}
{s}_j' & j\in[\ell]
\\
{s}_j & otherwise
\end{cases}.\]
It is trivial to see $\I^{0}=\I$, $\I^{n}=\I'$, $\s={\s}^{0}$ and $\s'={\s}^{n}$.
Recall for $j\in [n]-B$, ${I}_j(\x_j)={I}'_j(\x_j)$, $s_j=s_j'$ and for $j\in B$, $d({I}_j(\x_j),{I}_j'(\x_j))=1$ and $s_j'=s_j+1$.

Therefore, for $\ell\notin B$, $\I^{\ell-1}=\I^{\ell}$ and $\s^{\ell-1}=\s^{\ell}$. That is,
\begin{equation*}
\label{eq:lm:inter_for_proof_desir_pro_1}
\sum_{E'\subseteq E}\left(T_{E-E'}(\I^{\ell-1})\prod_{j\in E'}s_j^{\ell-1}\right)= \sum_{E'\subseteq E}\left(T_{E-E'}(\I^{\ell})\prod_{j\in E'}s_j^{\ell}\right).
\end{equation*}

For $\ell\in B$, $\I^{\ell-1},\I^{\ell}$ differ by one tuple in $R_{\ell}(\x_{\ell})$ and $\s^{\ell-1},\s^{\ell}$ only differ by $s^{\ell} = s^{\ell-1}+1$. Based on Lemma~\ref{lm:base_for_proof_desir_pro}, we have 
\begin{equation*}
\label{eq:lm:inter_for_proof_desir_pro_2}
\sum_{E'\subseteq E}\left(T_{E-E'}(\I^{\ell-1})\prod_{j\in E'}s_j^{\ell-1}\right)\leq \sum_{E'\subseteq E}\left(T_{E-E'}(\I^{\ell})\prod_{j\in E'}s_j^{\ell}\right).
\end{equation*}

Above all, for any $\ell\in[n]$, 
\begin{equation*}
\sum_{E'\subseteq E}\left(T_{E-E'}(\I^{\ell-1})\prod_{j\in E'}s_j^{\ell-1}\right)\leq \sum_{E'\subseteq E}\left(T_{E-E'}(\I^{\ell})\prod_{j\in E'}s_j^{\ell}\right).
\end{equation*}
And finally, 
\begin{align*}
\sum_{E'\subseteq E}\left(T_{E-E'}(\I)\prod_{j\in E'}s_j\right)=&\sum_{E'\subseteq E}\left(T_{E-E'}(\I^{0})\prod_{j\in E'}s_j^{0}\right)
\\
\leq& \sum_{E'\subseteq E}\left(T_{E-E'}(\I^{n})\prod_{j\in E'}s_j^{n}\right)
\\
=& \sum_{E'\subseteq E}\left(T_{E-E'}(\I')\prod_{j\in E'}s_j'\right)
\end{align*}
\end{proof}

With Lemma~\ref{lm:inter_for_proof_desir_pro}, we can show the smoothness property of $\hat{LS}^{(k)}(\cdot)$ for CQs with self-joins.

\begin{theorem}
\label{th:deseired_properity}
For any CQ and any $\mathbf{I},\mathbf{I}'$ such that $d(\mathbf{I},\mathbf{I}')=1$,
$\hat{LS}^{(k)}(\mathbf{I})\leq \hat{LS}^{(k+1)}(\mathbf{I}')$ for any $k\ge 0$.
\end{theorem}

\begin{proof}
Based on (\ref{eq:hat_lsk}), we have
\[\begin{cases}
\hat{LS}^{(k)}(\mathbf{I})=\max_{\mathbf{s}\in \mathbf{S}^k}\max_{i\in P^m}\sum_{E\subseteq D_i,E\neq \emptyset}\hat{T}_{[n]-E,\s}(\I);
\\
\hat{LS}^{(k+1)}(\mathbf{I}')=\max_{\mathbf{s}\in \mathbf{S}^{k+1}}\max_{i\in P^m}\sum_{E\subseteq D_i,E\neq \emptyset}\hat{T}_{[n]-E,\s}(\I').
\end{cases}\]
Recall formulation of $\hat{E}_{E,\s}(\I)$ in (\ref{eq:hat_TES}). It is sufficient to show, for any $\s\in \mathbf{S}^k$, we can find a $\s'\in\mathbf{S}^{k+1}$ such that for any $E\subseteq[n]$
\begin{equation}
\label{eq:th:deseired_properity_1}
\sum_{E'\subseteq E}\left(T_{E-E'}(\I)\prod_{j\in E'}s_j\right)\leq
\sum_{E'\subseteq E}\left(T_{E-E'}(\I')\prod_{j\in E'}s_j'\right)
\end{equation}

Assume for $\I,\I'$ differs by tuple in relation $R_{l_i},i\in P^m$. That is, for all $j\in D_i, d(I_j(\x_j),I_j'(\x_j))=1$. Given $\s\in \mathbf{S}^k$, we construct a $\s'$ such that
\[{s}_{j}'=\begin{cases}{s}_{j}+1&j\in D_i\\
{s}_{j}&otherwise\end{cases}.\]
It is clear $\s'\in \mathbf{S}^{k+1}$ and we show $\s'$ meets the requirement in (\ref{eq:th:deseired_properity_1}).

We follow the idea in the proof of Lemma~\ref{lm:sensitivity_TE}: regard the $q$ as a self-join-free query on the logical instance and any $\I,\I'$ with $d(\I,\I')=1,d(I_{l_i},I_{l_i}')=1$ correspond to $\{I_j(\x_j)\}_j,\{I_j'(\x_j)\}_i$ with $d(I_j(\x_j),I_j(\x_j))=1$ for all $j\in D_i$. 
Besides, recall, $\s$ and $\s'$ only differ by ${s}_{j}'={s}_{j}+1$ for all $j\in D_i$. 

Finally, we derive, 
\begin{align*}
\sum_{E'\subseteq E}\left(T_{E-E'}(\I)\prod_{j\in E'}s_j\right)=&\sum_{E'\subseteq E}\left(T_{E-E'}(\{I_j(\x_j)\}_j)\prod_{j\in E'}s_j\right)
\\
\leq&\sum_{E'\subseteq E}\left(T_{E-E'}(\{I_j'(\x_j)\}_j)\prod_{j\in E'}s_j'\right)
\\
=&\sum_{E'\subseteq E}\left(T_{E-E'}(\I')\prod_{j\in E'}s_j'\right),
\end{align*}
where the inequality follows from Lemma~\ref{lm:inter_for_proof_desir_pro}.
\end{proof}

\subsection{Computing \texorpdfstring{$RS(\cdot)$}{TEXT}}
\label{sec:rs_cq}

Recall that for any given $k$, $\hat{LS}^{(k)}(\I)$ can be computed in polynomial time since each $T_{\bar E}(\I)$ is an AJAR/FAQ query \cite{joglekar2016ajar,abo2016faq}. The last missing piece is to bound the range of $k$ that one has to consider when computing $RS(\I)$ using \eqref{eq:rs}.  The following lemma implies that we only need to consider $k=0,\dots, \hat{k}=O(1)$ when computing $RS(\I)$.  

\begin{lemma}
\label{lm:max_k}
For any $k\geq  \hat{k} =\frac{m_p}{1-\exp(-\beta/\max_{i\in[m]}n_i)}$, \[e^{-\beta k}\hat{LS}^{(k)}(\I)\leq e^{-\beta (k-1)}\hat{LS}^{(k-1)}(\I).\]
\end{lemma}

\begin{proof}
We expand $e^{-\beta k}\hat{LS}^{(k)}(\I)$ with (\ref{eq:hat_lsk}), (\ref{eq:hat_TES})
\[e^{-\beta k}\hat{LS}^{(k)}(\I) =\max_{\s\in \mathbf{S}^{k}} \max_{i\in P^m}\sum_{E\subseteq D_i, E\neq \emptyset}\sum_{E'\subseteq [n]-E}\left(e^{-\beta k}T_{[n]-E-E'}\prod_{j\in E'}s_j\right).\]

Now, we show when $k\geq\frac{m_p}{1-\exp(-\beta/\max_{i\in[m]}n_i)}$, for any $\s\in \mathbf{S}^k$ there is an $\s'\in \mathbf{S}^{k-1}$ such that for any $i\in P^m$,
\begin{align}
&\sum_{E\subseteq D_i, E\neq \emptyset}\sum_{E'\subseteq [n]-E}\left(e^{-\beta k}T_{[n]-E-E'}\prod_{j\in E'}s_j\right)
\nonumber
\\
\leq& \sum_{E\subseteq D_i, E\neq \emptyset}\sum_{E'\subseteq [n]-E}\left(e^{-\beta (k-1)}T_{[n]-E-E'}\prod_{j\in E'}s_j'\right).
\label{eq:lm:max_k_1}
\end{align}

Since $\s\in\mathbf{S}^k, k\geq\frac{m_p}{1-\exp(-\beta/\max_{i\in[m]}n_i)}$, there must exist one $i'\in[m]$, such that $s_{l_{i'}}\geq \frac{1}{1-\exp(-\beta/\max_{i\in[m]}n_i)}$. Then, we define the $\s'$ as
\[s'_{j}=\begin{cases}
s_{j}-1, &j\in D_{i'};
\\
s_{j}, & otherwise.
\end{cases}\]

Then, to prove (\ref{eq:lm:max_k_1}), it is sufficient to show for any $E'\subseteq E\subset [n]$,
\begin{equation}
e^{-\beta k}T_{E-E'}(\I)\prod_{j\in E'}s_j\leq e^{-\beta (k-1)}T_{E-E'}(\I)\prod_{j\in E'}s_j'.
\label{eq:lm:max_k_2}
\end{equation}
(\ref{eq:lm:max_k_2}) can be simplified to 
\[\prod_{j\in E'}s_j\leq e^{\beta}\prod_{j\in E'}s_j',\]
which can be further reduced to 
\[\prod_{j\in E'\cap D_{i'}}s_j\leq e^{\beta}\prod_{j\in E'\cap D_{i'}}s_j',\]
since $s_j=s_j'$ for $j\notin D_{i'}$.

Since
\[\prod_{j\in E'\cap D_{i'}}\frac{s_j}{s_j'}\leq \left(\frac{1}{1-\frac{1}{k}}\right)^{\max_{i\in[m]}n_i}\leq (e^{\beta/\max_{i\in[m]}n_i})^{\max_{i\in[m]}n_i}=e^{\beta},\]
we prove the claim.
\end{proof}

\paragraph{Remark}
In actual implementation, we first compute $T_{\bar E}(\I)$ for all $E\subseteq D_i, E\ne \emptyset$.  After that, it  only takes $O(1)$ time to compute $RS(\I)$ using formulas \eqref{eq:hat_lsk}, \eqref{eq:hat_TES}, and \eqref{eq:rs}. Thus, the concrete computational complexity of $RS(\cdot)$ for a CQ $q$ is $O(N^{w_{\max}})$, where $w_{\max}$ is the maximum AJAR/FAQ width \cite{joglekar2016ajar,abo2016faq} of the residual queries of $q$. 
  


\section{Neighborhood Optimality}
\label{sec:optimality_analysis}

In this section we prove Theorem \ref{th:fullCQ}. This is done in three steps: We first derive a sufficient condition for $SS(\cdot)$ to be an $r$-neighborhood lower bound.  Next, we show that this condition holds for full CQs with $r=O(1)$.  Finally, we show that $RS(\cdot)$ is at most a constant factor larger than $SS(\cdot)$, 

\subsection{General Neighborhood Lower Bounds}
\label{sec:general_lower_bound}
We first develop two general neighborhood lower bounds that hold for arbitrary queries (not necessarily CQs), one based on $LS^{(k)}(\cdot)$ while the other based on $SS(\cdot)$.  These lower bounds hold for an arbitrary query $q$ with vectored outputs.
We start with an observation from \cite{vadhan2017complexity}: 

\begin{lemma}[\cite{vadhan2017complexity}]
\label{lm:ls_origin_bound}
For any $\varepsilon$-DP mechanism $\mathcal{M}'(\cdot)$ and any instance $\I$, there exists an $\I'$ with $d(\I,\I')\leq 1$ such that
\begin{equation*}
\label{eq:ls_origin}
\Pr\left[\left|\mathcal{M}'(\I')-q(\I')\right|\geq \frac{LS(\I)}{2}\right]\geq \frac{1}{1+e^{\varepsilon}}.
\end{equation*}
\end{lemma}

This implies that $LS(\cdot)$ is an $1$-neighborhood lower bound, i.e., 
\begin{equation} 
\label{eq:lslb}
\max_{\I': d(\I,\I')\le 1} \mathrm{Err}(\mathcal{M}',\I') \geq \frac{1}{2 \sqrt{1+e^{\varepsilon}}}\cdot LS(\I),
\end{equation}
for any $\I$ and any $\mathcal{M'}$. We generalize this result, showing that $LS^{(r-1)}(\cdot)$ is an $r$-neighborhood lower bound:

\begin{lemma}
\label{lm:lb_lower_bound}
For any $\I$, any $\varepsilon$-DP mechanism $\mathcal{M}'$, and any $r\geq 1$,
\begin{equation}
\label{eq:lsrlb}
\max_{\I': d(\I,\I')\le r} \mathrm{Err}(\mathcal{M}',\I') \ge \frac{1}{2\sqrt{1+e^{\varepsilon}}}\cdot LS^{(r-1)}(\I).
\end{equation}
\end{lemma}

\begin{proof}
For any $\I$, we need show that there exists an $\I'$ in its $r$-neighborhood such that 
\[
 \mathrm{Err}(\mathcal{M}',\I') \geq \frac{1}{2\sqrt{1+e^{\varepsilon}}} \cdot LS^{(r-1)}(\I).
\]

For any $r\geq 1$, let
\[\I^* = \argmax_{\I',d(\I,\I')\leq r-1}LS(\I').\]
By Lemma~\ref{lm:ls_origin_bound}, for any $\varepsilon$-DP mechanism $\mathcal{M}'(\cdot)$, there exists an $\I'$ with $d(\I^*,\I')\leq 1$ such that
\[
\mathsf{E}\left[(\mathcal{M}'(\I')-|q(\I')|)^2\right] \geq \frac{1}{1+e^{\varepsilon}}\cdot \left({LS(\I^{*}) \over 2}\right)^2,
\]
so
\[
 \mathrm{Err}(\mathcal{M}',\I') \geq \frac{1}{2\sqrt{1+e^{\varepsilon}}} \cdot LS(\I^*) = \frac{1}{2\sqrt{1+e^{\varepsilon}}} \cdot LS^{(r-1)}(\I).
\]

Since $d(\I, \I^*)\le r-1$ and $d(\I^*,\I')\leq 1$, we have $d(\I,\I')\leq r$, i.e., $\I'$ is in the $r$-neighborhood of $\I$, as desired. 
\end{proof}
Note that \eqref{eq:lslb} is the special case of \eqref{eq:lsrlb} with $r=1$. 

\paragraph{Remark} 
Previously, Asi and Duchi~\cite{asi2020instance} also derive a neighborhood lower bound.  Here, we show that our lower bound is always no worse than theirs for $\varepsilon=O(1)$, while can be polynomially better for certain queries.  Furthermore, their lower bound requires a technical condition on $q$ while our lower bound holds for an arbitrary $q$. 

The $r$-neighborhood lower bound by Asi and Duchi (Lemma C.1 in \cite{asi2020instance}), when specialized to the 1-dimensional case, is as follows. Given a query $q$, for any $k$ and any instance $\I$, define 
\begin{equation*}
\label{eq:omega_I_r}
\omega(\I,r) := \max_{\I',d(\I,\I')\leq k}|q(\I)-q(\I')|.
\end{equation*}
If $\{ q(\I') : d(\I,\I') \le k\}$ contains an interval of length $c \cdot \omega(\I, k)$ for some $c>0$ and all $k\le r$, then 
\begin{equation}
\label{eq:DuchiLB}\max_{\I': d(\I,\I')\le r} \mathrm{Err}(\mathcal{M}',\I') \ge {c \over 16} \max_{k\leq r}e^{-\varepsilon k}\omega(\I,k).
\end{equation}
Our lower bound, which does not require any condition on $q$, is
\begin{equation}
\label{eq:ourLB}
\max_{\I': d(\I,\I')\le r} \mathrm{Err}(\mathcal{M}',\I') \ge {1 \over 2\sqrt{1+e^\varepsilon}}\cdot LS^{(r-1)}(\I).
\end{equation}

Next we compare \eqref{eq:DuchiLB} and \eqref{eq:ourLB}. 
By the definition of $LS^{(k)}(\I)$ and $\omega(\I,k)$, we have
\begin{equation*}
\label{eq:compare_lower_bound_q1}
\omega(\I,k)\leq \sum_{0\leq k'\leq k-1}LS^{(k')}(\I)\leq k\cdot LS^{(k-1)}(\I).
\end{equation*}
Then, 
\[\eqref{eq:DuchiLB} \le {c \over 16} \max_{k\leq r} (e^{-\varepsilon k} \cdot k\cdot LS^{(k-1)}(\I)) \le  {c \over 16} \cdot \max_{k\leq r} (ke^{-\varepsilon k}) \cdot LS^{(r-1)}(\I),  \]
which is asymptotically upper bounded by \eqref{eq:ourLB} as long as 
\[k e^{-\varepsilon k} \le O\left({1 \over \sqrt{1+e^\varepsilon}}\right),\]
which is true when $\varepsilon=O(1)$. 

On the other hand, the gap between \eqref{eq:DuchiLB} and \eqref{eq:ourLB} can be $\mathrm{poly}(N)$.  Consider the {\sc Median} query with a constant $\varepsilon$. Let $\I$ consist of $\log N$ copies of $0.5$, while the remaining entries are half $0$ and half $1$.  For any $r\ge \log N$, our lower bound \eqref{eq:ourLB}  is $LS^{(r)}(\I)=0.5$, while their lower bound \eqref{eq:DuchiLB} is
\[\max_{k\leq r}e^{-\varepsilon k}\omega(\I,k) \leq e^{-\varepsilon\log N} \cdot 0.5 = 1/N^{\Omega(1)}.\]

Nevertheless, Lemma C.1 in \cite{asi2020instance} yields better lower bounds for high-dimensional queries.

To show that $SS(\cdot)$ is an $r$-neighborhood lower bound, we need a condition on $LS^{(k)}(\cdot)$, that they do not grow more than exponentially quickly when $k\ge r$. 

\begin{lemma}
\label{lm:sufficient_condition_ss_neighborhood_optimal}
Given any $\varepsilon, \beta>0$ and any instance $\I$, if for some $r\ge 1$ (possibly depending on $\beta$ and $\I$),
\begin{equation}
\label{eq:lm:sufficient_condition_ss_neighborhood_optimal}
LS^{(k)}(\I)\leq e^{\beta k}  LS^{(r-1)}(\I),
\end{equation}
 for any $k\geq r$, for any $\varepsilon$-DP mechanism $\mathcal{M}'$,
\[ \max_{\I': d(\I,\I')\le r} \mathrm{Err}(\mathcal{M}',\I') \ge \frac{1}{2 \sqrt{1+e^{\varepsilon}}}\cdot SS(\I). \]
\end{lemma}

\begin{proof}
Consider any $\I$. When $k<r$,
\[
e^{-\beta k}LS^{(k)}(\I)
\leq LS^{(k)}(\I) \leq LS^{(r-1)}(\I),
\]
where the second inequality follows from the definition of $LS^{(k)}(\I)$ in (\ref{eq:lsk}).

When $k\geq r$, from the premise of the lemma, we have
\begin{equation*}
e^{-\beta k}LS^{(k)}(\I)\leq e^{-\beta k}e^{\beta k} LS^{(r-1)}(\I)=LS^{(r-1)}(\I).
\end{equation*}

Therefore, for any $k\geq 0$, $e^{-\beta k}LS^{(k)}(\I)\leq LS^{(r-1)}(\I)$. So
\begin{equation}
\label{eq:lm:sufficient_condition_ss_neighborhood_optimal_2}
SS(\I) = \max_{ k\geq 0}e^{-\beta k}LS^{(k)}(\I)\leq LS^{(r-1)}(\I).
\end{equation}

Finally, combine (\ref{eq:lm:sufficient_condition_ss_neighborhood_optimal_2}) and Lemma~\ref{lm:lb_lower_bound}, we prove the lemma.
\end{proof}

\paragraph{Remark 1} Recall from Section~\ref{sec:dp_mechanisms} that $\beta$ and $\varepsilon$ are just a constant-factor apart, so $\beta$ is also a constant if $\varepsilon$ is considered a constant.  

\paragraph{Remark 2} The restriction of the growth rate is very mild, except that it also forbids $LS^{(k)}(\cdot)$ to go from zero to nonzero.  This is why we impose this restriction only for $k\ge r$.  For certain problems like {\sc Median}, this requires a large $r$, which is actually unavoidable since a large flat neighborhood will rule out $r$-neighborhood optimal mechanisms for small $r$ anyway, as we argued in Section~\ref{sec:intro}.

\medskip 
Before considering CQs, as a warm-up, we apply Lemma~\ref{lm:sufficient_condition_ss_neighborhood_optimal} to the triangle counting problem.  Here, the instance $\I$ is a simple graph (i.e., no self-loops and multi-edges), and the query $q$ returns the number of triangles in $\I$.  

\begin{lemma}
\label{lm:optimality_ss_in_triangle_counting_problem}
For the triangle counting problem, the condition in Lemma \ref{lm:sufficient_condition_ss_neighborhood_optimal} holds with $r=\max\left\{3,\left \lceil 4\frac{\ln(1/\beta)}{\beta}\right \rceil\right\}$. 
\end{lemma}

\begin{proof}
Let $V$ be the domain of vertices.  For $i,j\in V$, let $x_{i,j}(\I)=1$ if there is an edge between vertex $i$ and $j$ in $\I$, and $0$ otherwise. Then the degree of vertex $i$ in $\I$ is
\[d_i(\I) = \sum_{j\in V}x_{i,j}(\I),\] and the number of common neighbors of vertices $i,j$ is
\[a_{i,j}(\I) = \sum_{k\in V}x_{i,k}(\I)\cdot x_{j,k}(\I).\]

An exact formula for $LS^{(k)}(\I)$, hence $SS(\I)$, is given in \cite{nissim2007smooth} for the triangle counting problem, but we only need the following upper and lower bound on $LS^{(k)}$:
\begin{equation}
\nonumber
\label{eq:lm:optimality_ss_in_triangle_counting_problem_1}
\max_{i,j\in V}a_{i,j}(\I)+\frac{k-1}{2} \le LS^{(k)}(\I) \leq  \max_{i,j\in V}a_{i,j}(\I)+k.
\end{equation}

We will show that by setting $r=\max\left\{3,\left \lceil 4\frac{\ln(1/\beta)}{\beta}\right \rceil\right\}$, (\ref{eq:lm:sufficient_condition_ss_neighborhood_optimal}) holds for any $\I$ and $k\ge r$. Thus $SS(\cdot)$ is $O(1)$-neighborhood optimal.

For $k\geq r$, we have
\begin{align*}
LS^{(r-1)}(\I)e^{\beta k} \geq & \max_{i,j\in V}a_{i,j}(\I)e^{\beta k}+\frac{r-1}{2}e^{\beta k}
\\
\geq & \max_{i,j\in V}a_{i,j}(\I)+e^{\beta k}
\\
\geq & \max_{i,j\in V}a_{i,j}(\I)+k
\\
\geq & LS^{(k)}(\I).
\end{align*}
The first equality is by the lower bound on $LS^{(k)}(\I)$; the second inequality is because $e^{\beta k}\geq 1$ and $r\geq 3$; the third inequality $e^{\beta k} \ge k$ follows from $k\ge r\ge  4\frac{\ln(1/\beta)}{\beta}$ and some simple algebra; the last inequality is by the upper bound on $LS^{(k)}(\I)$. 
\end{proof}

Note that the $r$ needed in the lemma above is independent of $\I$.  Thus, $SS(\cdot)$ is an $O(1)$-neighborhood lower bound for the triangle counting problem, i.e., the previous $SS$-based DP-mechanism for triangle counting \cite{nissim2007smooth} is $O(1)$-neighborhood optimal. This is the first optimality guarantee for triangle counting under DP, while our main optimality result is a vast generalization of this. 

\subsection{Neighborhood Lower Bound for CQs}
\label{sec:optimality_ss}

To show that $SS(\cdot)$ is an $O(1)$-neighborhood lower bound for CQs, we need to show that the condition in Lemma~\ref{lm:sufficient_condition_ss_neighborhood_optimal} holds  with some constant $r$.  This requires an upper bound on $LS^{(k)}(\cdot)$, as well as a lower bound on $LS^{(r-1)}(\cdot)$. For the upper bound on $LS^{(k)}(\cdot)$, we use the $\hat{LS}^{(k)}(\cdot)$ developed in Section~\ref{sec:hat_lsk}.  For the lower bound, we first consider the case $r=n_P$. Recall that $n_P=|P^n|$ is the number of private logical relations.

\begin{lemma}
\label{lm:compare_TE_LSn}
For any CQ, any instance $\I$, and any $E\subseteq P^n$, $E\neq \emptyset$, we have $LS^{(n_P-1)}(\I)\geq T_{\bar{E}}(\I)$.
\end{lemma}

\begin{proof}
We will construct an $\I'$ from $\I$ such that $d(\I,\I')\leq n_P-1$ and $LS(\I')\geq T_{\bar{E}}(\I)$. Recall $t_{\bar{E}}(\I)$ is the witness of $T_{\bar{E}}(\I)$. Then, we write $T_{\bar{E}}(\I)$ as
\[T_{\bar{E}}(\I)=|\Join_{i\in {\bar{E}}}(I_i(\x_i)\ltimes t_{\bar{E}}(\I))|.\]

We construct the $\I'$ as follows.
First, we fix $j_E$ and $i_E$ such that $j_E\in E$ and $j_E\in D_{i_E}$. 
Then, we find a tuple $t'\in \dom{(\cup_{i\in E}\x_i)}$ such that $\pi_{\partial q_{{E}}}t'=t_{\bar{E}}(\I)$. 
Next, for each $j\in E-\{j_E\}$, we add $\pi_{\x_j}t'$ to $I_j$ unless it already exists in $I_j$. 
Here, we at most add $|E|-1$ tuples. 
Since $E\subseteq P^n$, $d(\I,\I')\leq n_P-1$ and $LS(\I')\leq LS^{(n_P-1)}(\I)$. 
Besides, we can ensure for each $j\in E-\{j_E\}$, $I_j'(\x_j)$ contains $\pi_{\x_j}t'$.

Therefore, it suffices to show $LS(\I')\geq T_{\bar{E}}(\I)$.
To do that, we will show by flipping the tuple $t_{j_E}=\pi_{\x_{j_E}}t'$ in $I'_{j_E}$, $q(\I')$ will change by at least $T_{\bar{E}}(\I)$. 
Suppose $t_{j_E}\notin I'_{j_E}$. 
The number of tuples changes by
\begin{align}
&|(\Join_{j\in D_{i_E}}(I_j'(\x_j)\cup t_{j_E}))\Join(\Join_{j\in [n]-D_{i_E}}I_j'(\x_j))
\nonumber
\\
&-(\Join_{j\in [n]}I_j'(\x_j))|
\nonumber
\\
\geq& |(I_{j_E}'(\x_j)\cup t_{j_E})\Join(\Join_{j\in [n]-j_E}I_j'(\x_j))
\nonumber
\\
&-(\Join_{j\in [n]}I_j'(\x_j))|
\nonumber
\\
=&|t_{j_E}\Join (\Join_{j\in [n]-j_E}I_j'(\x_j))|
\nonumber
\\
=&|t_{j_E}\Join (\Join_{j\in E-j_E}I_j'(\x_j))\Join (\Join_{[n]-E}I_j'(\x_j))|
\nonumber
\\
\geq &|\pi_{\x_{j_E}}t' \Join (\Join_{j\in E-j_E}\pi_{\x_j}t')\Join (\Join_{j\in\bar{E}}I_j(\x_j))|
\label{eq:lm:compare_TE_LSn_1}
\\
= &|t'\Join (\Join_{\bar{E}}I_j(\x_j))|
\nonumber
\\
= &|\Join_{i\in {\bar{E}}}(I_i(\x_i)\ltimes t_{\bar{E}}(\I))|=T_{\bar{E}}(\I)
\nonumber
\end{align}
The (\ref{eq:lm:compare_TE_LSn_1}) is because for each $j\in E-\{j_E\}$, $I_j'(\x_j)$ contains $\pi_{\x_j}t'$ and for $j\in[n]-E$, $I_j(\x_j)=I_j'(\x_j)$.
For the case $t_{j_E}\in I'_i$, we can draw a similar conclusion.
\end{proof}

Next, recall from equations (\ref{eq:hat_lsk}) and (\ref{eq:hat_TES}) that $\hat{LS}^{(k)}(\I)$ is also defined in terms of the $T_{\bar E}(\I)$'s.  Together with Lemma~\ref{lm:compare_TE_LSn}, this allows us to build a connection between $\hat{LS}^{(k)}(\I)$ and $LS^{(n_P-1)}(\I)$:

\begin{lemma}
\label{lm:hat_lsk_bound}
For any CQ, any instance $\I$, and any $k\ge 1$, we have
\[\hat{LS}^{(k)}(\I)\leq (4k)^{n_P-1}LS^{(n_P-1)}(\I).\]
\end{lemma}

\begin{proof}
Based on (\ref{eq:hat_lsk}) and (\ref{eq:hat_TES}), we have
\begin{align}
\hat{LS}^{(k)}(\I)=&\max_{\mathbf{s}\in \mathbf{S}^{k}}\max_{i\in P^m} \sum_{E\subseteq D_i,E\neq \emptyset} \sum_{E'\subseteq[n]-E}\left(T_{[n]-E-E'}(\I)\prod_{j\in E'}s_j\right)
\nonumber
\\
=&\max_{\mathbf{s}\in \mathbf{S}^{k}}\max_{i\in P^m} \sum_{E\subseteq D_i,E\neq \emptyset} \sum_{E'\subseteq P^n-E}\left(T_{[n]-E-E'}(\I)\prod_{j\in E'}s_j\right)
\label{eq:lm:hat_lsk_bound_1}
\end{align}
The (\ref{eq:lm:hat_lsk_bound_1}) is because for any $\mathbf{s}$, $s_j=0$ for any $j\notin P^n$.

For above $E$ and $E'$, $E\cup E'\subseteq P^n$. 
Based on Lemma~\ref{lm:compare_TE_LSn}, we have $T_{[n]-E-E'}(\I)\leq LS^{(n_P-1)}(\I)$.
Besides, for any $\s\in\mathbf{S}^k$, since $E'\subseteq P^n-E,E\neq\emptyset$, we have $\prod_{j\in E'}s_j\leq k^{n_P-1}$.
Plug these into (\ref{eq:lm:hat_lsk_bound_1}),

\begin{align*}
\hat{LS}^{(k)}(\I)
\leq&\max_{\mathbf{s}\in\mathbf{S}^{k}}\max_{i\in P^m}\sum_{E\subseteq D_i,E\neq \emptyset}
\sum_{E'\subseteq P^n-E}\left(LS^{(n_P-1)}(\I)k^{n_P-1}\right)
\\
\leq&\max_{\mathbf{s}\in\mathbf{S}^{k}}\max_{i\in P^m}\sum_{E\subseteq D_i,E\neq \emptyset}
\left(2^{n_P-1}LS^{(n_P-1)}(\I)k^{n_P-1}\right)
\\
\leq&\max_{\mathbf{s}\in\mathbf{S}^{(k)}}\max_{i\in[m]}
(2^{n_P-1}2^{n_P-1}LS^{(n_P-1)}(\I)k^{n_P-1})
\\
=&(4k)^{n_P-1}LS^{(n_P-1)}(\I)
\end{align*}
\end{proof}

Lemma ~\ref{lm:hat_lsk_bound} almost meets the condition of Lemma~\ref{lm:sufficient_condition_ss_neighborhood_optimal}, except that $(4k)^{n_P-1}$ is not necessarily smaller than $e^{\beta k}$.  But as the former is a polynomial while the latter is exponential, this is not an issue as long as $k$ is larger than a constant. 

\begin{theorem}
\label{th:ss_lower_bound}
For any CQ $q$, any $\varepsilon, \beta>0$, there exist a constant $r>0$ (depending on $q$, $\varepsilon, \beta$) such that for any $\I$ and any $\varepsilon$-DP mechanism $\mathcal{M}'$,
\[ \max_{\I': d(\I,\I')\le r} \mathrm{Err}(\mathcal{M}',\I') \ge \frac{1}{2 \sqrt{1+e^{\varepsilon}}}\cdot SS(\I). \]
\end{theorem}

\begin{proof}
We will show that
\begin{equation}
\label{eq:lskhatopt}
\hat{LS}^{(k)}(\I)\leq e^{\beta k}LS^{(r-1)}(\I),
\end{equation}
for all $k\geq  \max\left\{4,n_P,\left \lceil  \frac{2(n_P-1)}{\beta}\ln \frac{2(n_P-1)}{\beta}\right \rceil\right\}$.

By Lemma~\ref{lm:hat_lsk_bound}, for $k\ge 4$, we have
\[\hat{LS}^{(k)}(\I)\leq (4k)^{n_P-1}LS^{(n_P-1)}(\I) \le k^{2(n_P-1)}LS^{(n_P-1)}(\I).\]

By setting
\[r:=\max\left\{4,n_P,\left \lceil  \frac{2(n_P-1)}{\beta}\ln \frac{2(n_P-1)}{\beta}\right \rceil\right\},\]
We can show that \eqref{eq:lskhatopt} holds for any $k\ge r$:
\begin{align*}
\hat{LS}^{(k)}(\I)\leq & k^{2(n_P-1)}LS^{(n_P-1)}(\I)
\\
\leq & e^{\beta k}LS^{(r-1)}(\I).
\end{align*}
The second inequality follow because when $r\geq n_P$, we have $LS^{(n_P-1)}(\I) \le LS^{(r-1)}(\I)$, and when $k\geq \frac{2(n_P-1)}{\beta}\ln \frac{2(n_P-1)}{\beta}$, we have $k^{2(n_P-1)}\leq e^{\beta k}$.
\end{proof}

\subsection{Optimality of \texorpdfstring{$RS(\cdot)$}{TEXT}}
\label{sec:optimality_rs}

To complete the proof of Theorem \ref{th:fullCQ}, we show that $RS(\cdot)$ is at most a constant-factor larger than $SS(\cdot)$.

\begin{lemma}
\label{lm:compare_RS_SS}
For any CQ and any $\I$, $RS(\I)\leq \left(\frac{4(n_P-1)}{\beta e^{1-\beta}}\right)^{n_P-1}SS(\I)$.
\end{lemma}

\begin{proof}
Recall the definition of $SS(\I)$ and $RS(\I)$ in (\ref{eq:ss}) and (\ref{eq:rs}),
\[SS(\I)=\max_{k\geq 0}e^{-\beta k}LS^{(k)}(\I),\]
\[RS(\I)=\max_{k\geq 0}e^{-\beta k}\hat{LS}^{(k)}(\I).\]
Let $k^*=\argmax e^{-\beta k}\hat{LS}^{(k)}(\I)$, and define the function 
\[g(k)=e^{-\beta k}(4k)^{n_P-1}LS^{(n_P-1)}(\I).\]
Setting its derivative to $0$, we see that $g(k)$ maximizes at $k_{\max} = \frac{n_P-1}{\beta}$ (even allowing k to take fractional values). Thus
\begin{equation}
\label{eq:th:compare_RS_SS_1}
g(k)\leq g(\frac{n_P-1}{\beta}).
\end{equation}

Therefore, 
\begin{align*}
RS(\I)= & e^{-\beta k^*}\hat{LS}^{(k^*)}(\I)
\nonumber
\\
\leq &e^{-\beta k^*} (4k^*)^{n_P-1}LS^{(n_P-1)}(\I)
\\
\leq &e^{-(n_P-1)}\left(\frac{4(n_P-1)}{\beta}\right)^{n_P-1}LS^{(n_P-1)}(\I)
\\
\leq &\left(\frac{4(n_P-1)}{\beta e^{1-\beta}}\right)^{n_P-1} \max_{k\geq 0} e^{-\beta k}LS^{(k)}(\I)
\\
= &\left(\frac{4(n_P-1)}{\beta e^{1-\beta}}\right)^{n_P-1}SS(\I).
\end{align*}
The first inequality follows from Lemma~\ref{lm:hat_lsk_bound}, and the second inequality is by (\ref{eq:th:compare_RS_SS_1}).
\end{proof}

\subsection{Elastic Sensitivity}
\label{sec:es}

Elastic sensitivity~\cite{johnson2018towards}, denoted as $ES(\cdot)$, is the only other DP mechanism for CQs with self-joins.  It is also a version of $\hat{SS}(\cdot)$, but defined using a different $\hat{LS}^{(k)}(\cdot)$. 
For $i\in[n],\x\subseteq \x_i$, let $mf(\x,I_i(\x_i))$ be the \textit{maximum frequency} in $I_i(\x_i)$ on attributes $\x$, i.e., $mf(\x,I_i(\x_i))=\max_{t\in\dom(\x)}|I_i(\x_i)\ltimes t|$. For $ES(\cdot)$, $\hat{LS}^{(k)}(\cdot)$ is defined as a product of a number of maximum frequencies; please see \cite{johnson2018towards} for the exact formula.

We give an example below showing that $ES(\I)$ can be asymptotically larger than $GS$. This means that $ES(\cdot)$ is not even worst-case optimal (i.e., not $N$-neighborhood optimal).

\begin{example}
Consider the path-4 query:
\[q=\mathtt{Edge}(x_1,x_2)\Join \mathtt{Edge}(x_2,x_3)\Join \mathtt{Edge}(x_3,x_4)\Join \mathtt{Edge}(x_4,x_5).\]
We showed that $GS =  O(N^2)$ in Example~\ref{ex:path4}.  Now consider the following instance $\I$ on $\mathtt{Edge}$ relation (assume the domain is $\mathbb{N}$):
\begin{align*}
\I(\mathtt{Edge}) = & \{(0, 1),(0,2),\dots,(0,{\frac{N}{2}}),
\\
&({\frac{N}{2}+1},{N+1}),\dots,({N},{N+1})\}.
\end{align*}
Note that $mf(x_i,E(x_i,x_{i+1})) = mf(x_{i+1},E(x_i,x_{i+1}))=\frac{N}{2}$ for $i=1,2,3,4$.  By the formula in \cite  {johnson2018towards}, we have $\hat{LS}^{(0)}(\I)= 4(\frac{N}{2})^3 = \frac{N^3}{2}$, thus
\[ES(\I) = \max_{k\geq 0}e^{-\beta k}\hat{LS}^{(k)}(\I)\geq \hat{LS}^{(0)}(\I) =\Omega(N^3).\]\end{example}

\begin{figure*}[htbp]
\includegraphics[width=1\textwidth]{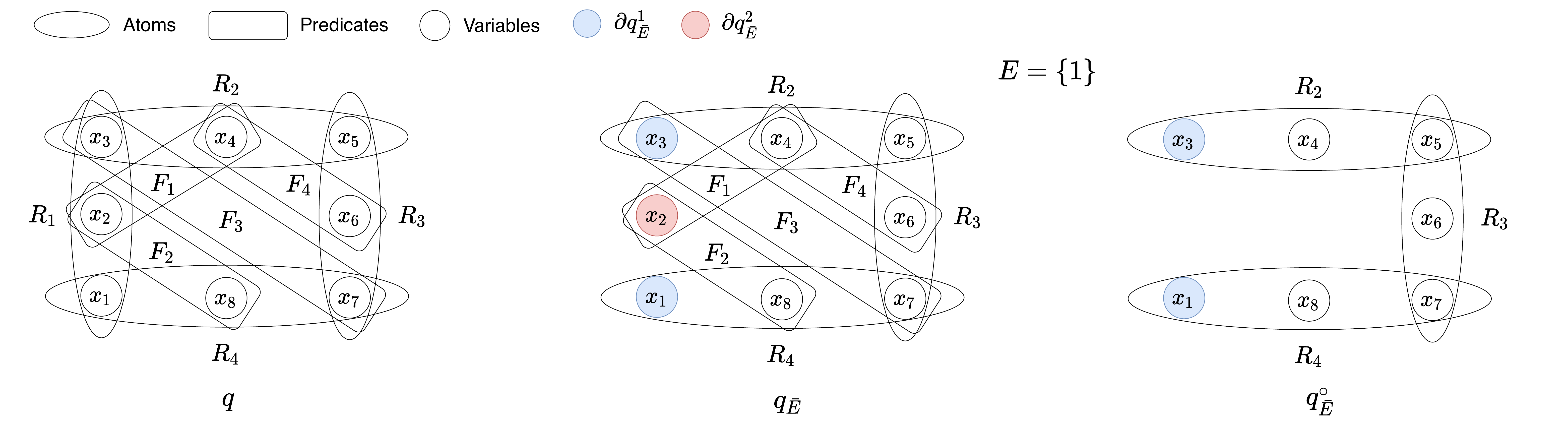}
    \caption{$q_{\bar{E}}$, $q_{\bar{E}}^{\circ}$, $\partial q_{\bar{E}}^1$ and $\partial q_{\bar{E}}^2$.}
    \label{fig:boundaries}
\end{figure*}

\section{CQ\lowercase{s} with Predicates}
\label{sec:rs_filter}

A CQ with predicates (CQP) has the form
\[
q:=\sigma_{P_1(\y_1)\wedge \cdots \wedge P_\kappa(\y_\kappa)} (R_1(\x_1)\Join \cdots \Join R_n(\x_n)),
\]
where each $P_j:\mathbf{dom}(\mathbf{y}_j) \rightarrow\{\mathsf{ True, False}\}$ is a computable function for some $\mathbf{y}_j \subseteq var(q) = \x_1\cup \cdots \cup \x_n$. 
By a slight abuse of notation, we also use $P(\mathbf{y})$ to denote the (possibly infinite) relation $\{ t \in \dom(\y) \mid P(t)\}$.  This way, a CQP can be written as a normal CQ:
\begin{equation}
\label{eq:q_filter_1}
q:=R_1(\x_1)\Join \cdots \Join R_n(\x_n)\Join P_1(\y_1)\Join \cdots \Join P_\kappa(\y_\kappa).
\end{equation}
Note that the $P_j(\y_j)$'s are all public, since they only depend on the query and the domain, not on the instance. 

The current approach to dealing with a CQP under DP \cite{johnson2018towards,kotsogiannis2019privatesql,dong21:residual} is to evaluate the CQP as given, but compute the sensitivity without considering the predicates.  This yields a valid DP mechanism, but loses optimality. To see this, just consider an extreme case where a predicate always returns $\mathsf{False}$.  Then the query becomes a trivial query and the optimal (under any notion of optimality) mechanism is $\mathcal{M}(\cdot)\equiv 0$, i.e., $\mathrm{Err}(\mathcal{M}, \I)=0$ for all $\I$, but the sensitivity of the query without the predicate must be nonzero.

In this section, we show how to extend $RS(\cdot)$ to CQPs while preserving its $O(1)$-neighborhood optimality.  The idea is simple, we just consider a CQP as a CQ as defined in \eqref{eq:q_filter_1}, so optimality immediately follows from Theorem \ref{th:fullCQ}.  The issue, however, is how to compute $RS(\cdot)$ when some relations are infinite.  In Section~\ref{sec:general_filters} we first give a general algorithm, which may take exponential time, to compute $RS(\cdot)$ for arbitrary predicates under the technical condition of Theorem \ref{th:filters}; in Section \ref{sec:inequality_comparison_filters} we give a polynomial-time algorithm for the case where all the predicates are inequalities or comparisons, which proves the second part of Theorem \ref{th:filters}.

\subsection{General Predicates}
\label{sec:general_filters}

The first observation is that, when the $P(\y_j)$'s are arbitrary (but still computable), it is undecidable to check if a given CQP is a trivial query.  Recall that if the query is trivial, the optimal DP mechanism $\mathcal{M}$ is deterministic and achieves $\mathrm{Err}(\mathcal{M}, \I)=0$ for all $\I$; otherwise, the mechanism must be probabilistic.  Since one cannot distinguish between the two cases, optimal (under any notion of optimality) DP mechanisms do not exist. For the undecidability result, just consider the simple CQP $q_M=R(x) \Join P_M(x)$, where $P_M(x)=\mathsf{True}$ iff the Turing machine $M$ terminates in less than $x$ steps. Note that $P_M(x)$ is decidable. However, it is easy to see that $|q_M(\cdot)|\equiv 0$ iff $M$ does not halt. 

However, the situation is not hopeless. Below we show how to compute $RS(\I)$ for any CQP if for any $\mathbf{z} \subseteq var(q)$, the satisfiability of $\varphi_1 \wedge \cdots \wedge\varphi_S$ is decidable, where each $\varphi_i$ is $P_j(\mathbf{u}_i)$ for any $j$ and $\mathbf{u}_i$ is $\mathbf{y}_j$ after replacing all variables not in $\mathbf{z}$ by any constants.  This is a very mild condition; in fact, the entire literature on \textit{constraint satisfaction problems (CSPs)} is devoted to designing efficient algorithms for determining the satisfiability of $\varphi_1 \wedge \cdots \wedge\varphi_S$ when the $\varphi_i$'s take certain forms, and finding a satisfying valuation for $\mathbf{z}$, if one exists.

It suffices to show how to compute $T_{\bar{E}}(\I)$ for any $E\subseteq P^n$. Since all the predicates correspond to public relations, the residual query has the form
\[q_{\bar{E}} = (\Join_{i\in \bar{E}} R_i(\x_i)) \Join (\Join_{j\in[\kappa]}P_j(\y_j)).\]
We split the boundary variables as $\partial q_{\bar{E}} = \partial q^1_{\bar{E}} \cup \partial q^2_{\bar{E}}$, where
\[\partial q_{\bar{E}}^{1} = \{x \mid x\in \x_i\cap \x_j,i\in E,j\in \bar{E}\},\]
and 
\[\partial q_{\bar{E}}^{2} = \{x \mid x\in \x_i\cap \y_j,i\in E,j\in[\kappa]\} - \partial q^1_{\bar E}.\]
Let 
\[q_{\bar{E}}^{\circ} = (\Join_{i\in \bar{E}} R_i(\x_i))\]
be the CQ part of $q_{\bar E}$. 

\begin{example}
\label{ex:boundaries}
Figure~\ref{fig:boundaries} illustrates these concepts with the query 
\begin{align*}
q=&R_1(x_1,x_2,x_3)\Join R_2(x_3,x_4,x_5)\Join R_3(x_5,x_6,x_7) \Join R_4(x_1,x_7,x_8)
\\
&\Join P_1(x_2,x_4)\Join P_2(x_2,x_8)\Join P_3(x_3,x_7)\Join P_4(x_4,x_6),
\end{align*}
where we set $E=\{1\}$.
\end{example}

The following observations about the boundary variables are straightforward.

\begin{lemma}
\label{lm:properties_boundary}
For any CQP $q$ and any $E\subseteq [n]$,
\begin{enumerate}
    \item $\partial q_{\bar{E}}^{2} \subseteq\mathbf{y}_1\cup \cdots \cup \y_\kappa \subseteq \partial q_{\bar{E}}^2 \cup var(q_{\bar{E}}^{\circ})$;
    \item $\partial q_{\bar{E}}^{1}\subseteq var(q_{\bar{E}}^{\circ})$;
    \item $var(q_{\bar{E}}^{\circ}) \cap \partial q_{\bar{E}}^{2} = \emptyset$. 
\end{enumerate}
\end{lemma}

Now, we look at how to compute $T_{\bar{E}}(\I)$:
\begin{align}
T_{\bar{E}}(\I) =&\max_{t\in\dom(\partial q_{\bar{E}})} |q_{\bar{E}}(\I)\Join t|
\nonumber
\\
=&\max_{t\in\dom(\partial q_{\bar{E}})} |q_{\bar{E}}^{\circ}(\I)\Join (\Join_{j\in[\kappa]}P_i(\y_i))\Join t|
\nonumber
\\
=&\max_{\substack{t_1\in \dom(\partial q_{\bar{E}}^1) \\ t_2\in \dom(\partial q_{\bar{E}}^2)}}  
\left| 
 q_{\bar{E}}^{\circ}(\I)\Join (\Join_{j\in[\kappa]} P_j(\y_j)) \Join t_1 
\Join t_2 \right|
\nonumber
\\
=&\max_{\substack{t_1\in \pi_{\partial q_{\bar{E}}^1}q_{\bar{E}}^{\circ}(\I)\\ t_2\in \dom(\partial q_{\bar{E}}^2)}}  
\left|
 q_{\bar{E}}^{\circ}(\I)\Join (\Join_{j\in[\kappa]} P_j(\y_j)) \Join t_1 
\Join t_2 \right|.
\nonumber
\end{align}
The last step is because only $t_1\in \pi_{\partial q_{\bar{E}}^1}q_{\bar{E}}^{\circ}(\I)$ can join with $q_{\bar{E}}^{\circ}(\I)$. Since $|q_{\bar{E}}^{\circ}(\I)|$ is bounded by $O(N^n)$, the choices of $t_1$ are limited. The difficulty is that $t_2\in \dom(\partial q_{\bar{E}}^2)$ has infinitely many choices.
The idea is to flip the problem around.  For any $B\subseteq q_{\bar{E}}^{\circ}(\I)$, we check if there exist $t_1, t_2$ such that 
\begin{equation}
\label{eq:checkB}
|B\Join (\Join_{j\in[\kappa]} P_j(\y_j)) \Join t_1 \Join t_2| = |B|.
\end{equation}
This is equivalent to checking if $t_B \Join t_1 \Join t_2$ can pass all predicates for every $t_B \in B$.  Since $\partial q_{\bar{E}}^{1}\subseteq var(q_{\bar{E}}^{\circ})$, for each $t_B\in B$, $t_1$ must be $\pi_{\partial q_{\bar{E}}^1}t_B$.  Thus the problem boils down to deciding if 
\begin{equation} 
\label{eq:sat}
\bigwedge_{t_B \in B, j\in [\kappa]} P_j(\y_j(t_B))
\end{equation}
is satisfiable, where $\y_j(t_B)$ denotes $\y_j$ after replacing its variables by the corresponding constants if they appear in $t_B$.  Note that free variables in \eqref{eq:sat} are $\mathbf{z}=\partial q^2_{\bar E}$. This is precisely the technical condition we impose on the predicates. Finally, we enumerate all $B$, and return the maximum $|B|$ for which \eqref{eq:sat} is satisfiable.  This proves the first part of Theorem~\ref{th:filters}.  However, this algorithm runs in exponential time since there are $2^{|q_{\bar{E}}^{\circ}(\I)|} = 2^{\mathrm{poly}(N)}$ $B$'s that need to be considered. 

\subsection{Comparison and Inequality Predicates}
\label{sec:inequality_comparison_filters}

For CQPs where the predicates are inequalities or comparisons, we may without loss of generality assume that the domain of all attributes in $\y_1\cup \cdots \cup \y_\kappa$ is $\fdom$.  We show in this subsection how to reduce the running time of the algorithm to $\mathrm{poly}(N)$ in this case. Let $\rho=|\partial q_{\bar{E}}^2|$.  Then $t_2$ takes values from $\fdom^\rho$. The key to an efficient algorithm is thus to reduce this domain, and then apply the algorithm in \cite{joglekar2016ajar,abo2016faq}.

To reduce the domain of $t_2$, one natural idea is to only consider the \textit{active domain}~\cite{abiteboul1995foundations}. Let $\adom(\I)$  be the set of integers appearing in $\I$ on attributes $\y_1\cup \cdots \cup \y_\kappa$, and let $\adom(q)$ be the set of integers appearing in the predicates of $q$.  Then the active domain is $\adom(q,\I)=\adom(q)\cup \adom(\I)\cup\{-\infty,\infty\}$. However, only considering $t_2\in (\adom(q,\I))^{\rho}$ is not enough as seen in the following example.

\begin{example}
Following Example~\ref{ex:boundaries}, suppose $P_1(x_2,x_4)$ is $x_2 > x_4$, $P_2(x_2,x_8)$ is $x_8>x_2$, while ignoring $P_3$ and $P_4$.  Consider the following instance $\I$:
\begin{align*}
R_1(x_1,x_2,x_3)&=\{(0,3,0),(0,5,0)\},\\
R_2(x_3,x_4,x_5)&=\{(0,1,0),(0,2,0),(0,3,0)\},
\\
R_3(x_5,x_6,x_7)&=\{(0,0,0)\},
\\
R_4(x_7,x_8,x_1)&=\{(0,5,0),(0,6,0),(0,7,0)\}.
\end{align*}
For $E=\{1\}$, $T_{\bar{E}}(\I)$ attains its maximum at $x_2=4$, which is not included in $\adom(q,\I)$.
\end{example}

This example shows that $T_{\bar E}(\I)$ may attain its maximum at some value between two consecutive values in the active domain.  Thus, we augment the active domain to $\audom(q,\I)$, as follows.  Let $\adom(q, \I, i)$ be the $i$th elements in $\adom(q, \I)$ in order.  $\audom(q,\I)$ includes all elements in  $\adom(q, \I, i)$, plus $2\kappa$ arbitrary distinct elements between $\adom(q, \I, i)$ and $\adom(q, \I, i+1)$ for all $i\in[|\adom(q,\I)|-1]$.  If there are less than $2 \kappa$ elements between $\adom(q, \I, i)$ and $\adom(q, \I, i+1)$, all elements in between are included.

We show that it suffices to use $\audom(q,\I)$ as the domain of $t_2$.

\begin{lemma}
\label{lm:filter_inequal_compare}
When all the predicates are inequalities and comparisons, 
\begin{equation}
\label{eq:Te_aug}
T_{\bar E}(\I) = \max_{\substack{t_1\in \pi_{\partial q_{\bar{E}}^1}q_{\bar{E}}^{\circ}(\I)\\ t_2\in (\audom(q,\I))^\rho}}  
\left|
 q_{\bar{E}}^{\circ}(\I)\Join (\Join_{j\in[\kappa]} P_j(\y_j)) \Join t_1 
\Join t_2 \right|.
\end{equation}
\end{lemma}

\begin{proof}
It is sufficient to show that, for any $t_2 \in \fdom^\rho$, we can find a $t_2' \in  (\audom(q,\I))^\rho$ such that 
\begin{align}
&|q_{\bar{E}}^{\circ}(\I)\Join (\Join_{j\in[\kappa]} P_j(\y_j)) \Join t_1 
\Join t_2|  \nonumber \\
=& |q_{\bar{E}}^{\circ}(\I)\Join (\Join_{j\in[\kappa]} P_j(\y_j)) \Join t_1 
\Join t_2'|. 
\label{eq:aug1}
\end{align}

We order the at most $\rho$ distinct values in $t_2$ as $v_1, v_2, \dots$, and map them to values in $\audom(q,\I)$ to obtain $t'_2$, as follows.  If $v_i \in \adom(q, \I)$, we map $v_i$ to itself. For values strictly between $\adom(q, \I, i)$ and $\adom(q, \I, i+1)$  for some $i$, we map them to the additional elements in $\audom(q,\I)$ between $\adom(q, \I, i)$ and $\adom(q, \I, i+1)$ in an order-preserving fashion. Since $\rho \le 2\kappa$, this is always possible. 

It is easy to see that (\ref{eq:aug1}) holds after the above mapping.  This is because the attributes of $t_2$ and $t'_2$ only appear in the predicates and all the comparison/inequality relationships remain unchanged.
\end{proof}

Since $\audom(q,\I)=O(N+\kappa) =O(N)$, we can simply materialize each $P_j(\y_j)$ into $\{ t \in (\audom(q,\I))^2\mid P_j(t)\}$, which has size $O(N^2)$.  Thus, evaluating \eqref{eq:Te_aug} using the algorithm in  \cite{joglekar2016ajar,abo2016faq} also takes polynomial time, and we have concluded the proof of the second part of Theorem \ref{th:filters}. 

As a practical improvement, observe that if a variable $y \in \partial q^2_{\bar E}$ is only involved in inequality predicates, then it can always take a value such that all these inequalities hold.  Thus, there is no need to materialize these predicates.  In particular, we arrive at a simpler formula for computing $T_{\bar E}(\I)$ when all predicates are inequalities, e.g., graph pattern counting queries.

\begin{corollary}
For a CQP where all predicates are inequalities, 
\[
T_{\bar E}(\I) = \max_{t_1\in \dom(\partial q_{\bar{E}}^1)}
\left|
 q_{\bar{E}}^{\circ}(\I)\Join (\Join_{j\in[\kappa], \y_j \subseteq var(q_{\bar E}^\circ)} P_j(\y_j)) \Join t_1 \right|.
\]
\end{corollary}
To compute $T_{\bar E}(\I)$, we compute $q_{\bar{E}}^{\circ}(\I)$, apply all predicates $P_j(\y_j))$ for $j\in[\kappa], \y_j \subseteq var(q_{\bar E}^\circ))$, do a count group-by $\partial q_{\bar{E}}^1$, and return the maximum count.

\section{Non-full CQ\lowercase{s}}
\label{sec:non_full_cq}

A non-full CQ has the form
\[q:=\pi_{\out}\left(R_1(\x_1)\Join \dots \Join R_n(\x_n)\right),\]
where $\out\subseteq\x$ denotes the set of output variables. 

Similarly, the current approach \cite{johnson2018towards,kotsogiannis2019privatesql,dong21:residual} simply computes the noise ignoring the projection.  This performs badly as the projection usually reduces the true count significantly, so the noise becomes relatively much larger. In this section, we show how to add projection into the  residual sensitivity framework.

For any $E\subseteq [n]$, define 
\[\out_E = \out\cap(\cup_{i\in E}\x_i).\]
Note that $\out = \out_{[n]}$.

Given $E\subseteq[n]$, the residual query with projection is
\[q_E:=\pi_{\out_E}(\Join_{i\in E}R_i(\x_i)).\]

The boundary variables $\partial q_E=\{x|x\in\x_i\cap\x_j,i\in E,j\in\bar{E}\}$ remain unchanged, but the maximum boundary $T_E(\I)$ and witness $t_E(\mathbf{I})$ are modified as 
\[T_{E}(\mathbf{I}) = \max_{t \in \dom(\partial q_E)}  |\pi_{\out_E}(\Join_{i\in E}I_i(\x_i)\Join t)|\]
and
\[t_{E}(\mathbf{I}) = \argmax_{t \in \dom(\partial q_E)} |\pi_{\out_E}(\Join_{i\in E}I_i(\x_i)\Join t)|.\]
If $\out_E=\emptyset$, there is always a $t\in \dom(\partial q_E)$ such that $(\Join_{i\in E}I_i(\x_i))\Join t \ne \emptyset$, which becomes $\{\langle \rangle\}$ after the projection, so $T_E(\I)=1$.

Note that these definitions degenerate into the full-CQ case when $\out = var(q)$. 

We as before compute $\hat{LS}^{(k)}(\I)$ by (\ref{eq:hat_lsk}), (\ref{eq:hat_TES}), and then $RS(\I)$ by (\ref{eq:rs}), but using the new definition of $T_E(\I)$ with projection.  Below we show that $RS(\cdot)$ is still a valid $\varepsilon$-DP mechanism and it can be computed efficiently.  Recall that the validity of $RS(\cdot)$ is based on (1) $\hat{LS}^{(k)}(\cdot)$ is an upper bound of ${LS}^{(k)}(\cdot)$; and (2) $\hat{LS}^{(k)}(\cdot)$ satisfies the smoothness property (\ref{eq:des_property}).   The first depends on Lemma~\ref{lm:sensitivity_TE} and Theorem~\ref{th:ls_self_join_query} while the second depends on Lemma~\ref{lm:sensitivity_TE}.  One can verify that, as long as Lemma~\ref{lm:sensitivity_TE} and Theorem~\ref{th:ls_self_join_query} hold for non-full CQs, the rest of the  validity proof will go through.  We thus focus on verifying  Lemma~\ref{lm:sensitivity_TE} and Theorem~\ref{th:ls_self_join_query} on non-full CQs.

\begin{lemma}
\label{lm:cliams_for_TE_non_full_cq}
For non-full CQs, Lemma~\ref{lm:simple_lemma}, and~\ref{lm:sensitivity_TE} still hold.
\end{lemma}

\begin{proof}
First, it is trivial to see, Lemma~\ref{lm:simple_lemma} is still valid. 
The validness of Lemma~\ref{lm:sensitivity_TE} is based on  the self-join-free version of Lemma~\ref{lm:sensitivity_TE}. 
Therefore, it suffices to prove the self-join-free version of Lemma~\ref{lm:sensitivity_TE}.

We first show, for non-full CQs without self-joins, given any $E\subseteq[n],i\in E$, and two instances $\mathbf{I}, \mathbf{I}'$ such that $d(I_i(\x_i), I'_i(\x_i))=1$, $I_j(\x_j) =I'_j(\x_j)$ for all $j \in E-\{i\}$, we have $|T_E(\mathbf{I})-T_E(\mathbf{I}')|\leq T_{E-\{i\}}(\mathbf{I})$.

There are three cases: $I_i'$ is obtained from insertion, deletion or change a tuple $t'$ from $I_i$. 
For the first case, if $\out_{E} = \emptyset$, $T_E(\I')=T_E(\I)=1$ thus
\[|T_E(\I')-T_E(\I)|=0\leq T_{E-\{i\}}(\I).\]

If $\out_{E} \neq \emptyset$, define
\[\tilde{T}_E(\I) = |\pi_{\y_E}((\Join_{i\in E}I_i(\x_i))\ltimes t_E(\I')) |\]
and have
\[|T_E(\I')-T_E(\I)|=T_E(\I')-T_E(\I)\leq T_E(\I')-\tilde{T}_E(\I).\]

Next, we bound $T_E(\I')-\tilde{T}_E(\I)$.
\begin{align}
&T_E(\I')-\tilde{T}_E(\I)
\nonumber
\\
\nonumber
=&|\pi_{\out_E}((\Join_{i\in E}I_i(\x_i))\ltimes t_E(\I'))|
-|\pi_{\out_E}((\Join_{i\in E}I_i'(\x_i))\ltimes t_E(\I'))|
\\
\nonumber
=& |\pi_{\out_E}((\Join_{j\in E-\{i\}} I_j(\x_j))\Join t'\ltimes t_E(\I'))|
\\
=& |\pi_{\out_{E-\{i\}}}((\Join_{j\in E-\{i\}} I_j(\x_j))\Join t'\ltimes t_E(\I'))|
\label{eq:lm:cliams_for_TE_non_full_cq_1}
\\
\leq& \max_{t\in \dom(\partial q_{E-\{i\}})}|\pi_{\out_{E-\{i\}}}((\Join_{j\in E-\{i\}} I_j(\x_j))\ltimes t)|
\label{eq:lm:cliams_for_TE_non_full_cq_2}
\\
=&T_{E-\{i\}}(\I).
\nonumber
\end{align}
(\ref{eq:lm:cliams_for_TE_non_full_cq_1}) is derived by fixing $t'$ will fix the values of attributes in $\x_i$. 
(\ref{eq:lm:cliams_for_TE_non_full_cq_2}) is because the attributes of $t'\Join t_E(\I')$ are divided into two parts, one is interior to $q_{[n]-(E-\{i\})}$ while the other is $\partial q_{E-\{i\}}$; the values of attributes in the first part will not affect the join between $t'\Join t_E(\I')$ and $\Join_{j\in E-\{i\}}I_j(\x_j)$.

For the second case, we can draw the conclusion with the same idea. For the third case, where $\I'$ is achieved by changing one tuple from $\I$, define their common part as $\I''$:
\begin{align}
&\begin{cases}
\left||T_{E}(\I)-T_{E}(\I'')|\right|\leq T_{E-\{i\}}(\I'')
\\
\left||T_{E}(\I')-T_{E}(\I'')|\right|\leq T_{E-\{i\}}(\I'')
\\
T_{E}(\I'),T_{E}(\I)\geq T_{E}(\I'')
\end{cases}
\nonumber
\\
\Rightarrow& \left||T_{E}(\I)|-|T_{E}(\I')|\right|\leq T_{E-\{i\}}(\I'')\leq T_{E-\{i\}}(\I).
\label{eq:th:ls_selection_projection_5}
\end{align}

By now, we prove for CQs without self-joins, any $\mathbf{I}, \mathbf{I}'$ such that they only differ by one tuple in $R_1(\x_i)$, the difference between $T_{E}(\I)$ and $T_{E}(\I')$ is at most $T_{E-\{i\}}(\I)$.
Then, we can follow the idea from Lemma 4.6 to 4.8 in our prior work~\cite{dong21:residual} to prove the self-join-free version of Lemma~\ref{lm:sensitivity_TE}.
\end{proof}

\begin{theorem}
\label{th:ls_non_full_cq}
For non-full CQs,
Theorem~\ref{th:ls_self_join_query} still holds.
\end{theorem}

\begin{proof}
As shown in the remark after Theorem~\ref{th:ls_self_join_query}, the validity of Theorem~\ref{th:ls_self_join_query} is based on Lemma~\ref{lm:sensitivity_TE} and~\ref{lm:ls_self_join_free_query}.
Since Lemma~\ref{lm:sensitivity_TE} has already been shown to be valid for non-full CQs, it suffices to only show Lemma~\ref{lm:ls_self_join_free_query} holds: for non-full CQs without self-joins,
$LS(\I)\leq \max_{i\in P^n}T_{[n]-\{i\}}(\I)$. Recall
\[LS(\I)=\max_{i\in P^n}\max_{\I':d(\I,\I')=1,d(I_i(\x_i),I'_i(\x_i))=1}\left||q(\I)|-|q(\I')|\right|.\]

Given $\I,\I'$ differing by one tuple in $R_i(\x_i)$. $\I'$ can be different from $\I$ in three ways: insertion, deletion or change a tuple $t'\in\dom(\x_i)$. In the first two cases
\begin{align}
&\left||q(\I)|-|q(\I')|\right|
\nonumber
\\
=&|\pi_{\out}(t'\Join (\Join_{j\in [n]-\{i\}}I_j(\x_j)))|
\nonumber
\\
=&|\pi_{\out_{[n]-\{i\}}}(t'\Join (\Join_{j\in [n]-\{i\}}I_j(\x_j)))|
\label{eq:th:ls_non_full_cq_1}
\\
\leq& \max_{t\in \dom(\partial q_{[n]-\{i\}})} |\pi_{\out_{[n]-\{i\}}}((\Join_{j\in [n]-\{i\}}I_j(\x_j))\ltimes t)|
\label{eq:th:ls_non_full_cq_2}
\\
=& T_{[n]-\{i\}}(\I).
\nonumber
\end{align} 
(\ref{eq:th:ls_non_full_cq_1}) is derived by fixing $t'$ will fix the values of attributes for $\x_i$. (\ref{eq:th:ls_non_full_cq_2}) is because $\x_i$ can be divided into two parts: $\partial q_{\{i\}}$ and $\x_i-\partial q_{\{i\}}$; the first part is equal to $\partial q_{[n]-\{i\}}$ while the second one does not affect the join.

In the case where $\I'$ is achieved by changing one tuple from $\I$, we can use a similar idea as Lemma~\ref{lm:cliams_for_TE_non_full_cq} and get $\left||q(\I)|-|q(\I')|\right|\leq T_{[n]-\{i\}}(\I)$.
\end{proof}

In terms of computation, we observe that $T_E(\I)$ with projection is still an AJAR/FAQ query, but now with 3 semiring aggregations ($\max, +, \max$), so it can still be computed by the algorithm in \cite{joglekar2016ajar,abo2016faq} in polynomial time.  Furthermore, one can verify that Lemma~\ref{lm:max_k} still holds non-full queries, so it takes $O(1)$ time to compute $RS(\cdot)$ after all the $T_E(\I)$'s have been computed.  

\begin{theorem}
For any non-full CQ $q$, $RS(\cdot)$ is an $\varepsilon$-DP mechanism that can be computed in $\mathrm{poly}(N)$ time.
\end{theorem}

Non-full CQs with predicates can be handled by combining the methods described in this and Section~\ref{sec:rs_filter}.  More precisely, for a non-full CQ with general predicates, we add $\pi_{\mathbf{o}_{\bar E}}$ on both sides of \eqref{eq:checkB}; if the predicates are inequalities and comparisons, we materialize each predicate and then apply the algorithm above. The resulting $RS(\cdot)$ is still $\varepsilon$-DP, and can be much smaller than that on the full CQ.  However, the lower bound Theorem~\ref{th:ss_lower_bound} no longer holds for non-full CQs, thus $RS(\cdot)$ is not $O(1)$-neighborhood optimal.  We complement this with the following negative result.

\begin{theorem}
\label{th:optimal_dp_mech_not_exist}
For any $\varepsilon>0$, any $(r,c)$-neighborhood optimal $\varepsilon$-DP mechanism $\mathcal{M}(\cdot)$ for the query $q:=\pi_{x_1}\left(R_1(x_1,x_2)\Join R_2(x_2)\right)$, where $R_1$ is the private relation, must have $c r^2\ge N$. 
\end{theorem}

\begin{proof}
Let $\mathcal{M}(\cdot)$ be as given. We will construct two instances $\I$ and $\I'$ such that $\mathcal{M}(\I)$ and $\mathcal{M}(\I')$ must differ a lot.  Suppose $\dom(x_1)=\dom(x_2) = \fdom$.  The public relation $R_2$ takes the same instance $I_2 = I'_2 = [r]$. For the private relation $R_1$, we set $I_1 = [N/r] \times [r]$ and $I'_1 = [N]\times \{0\}$. Note that $|q(\I)| = N/r$ and $|q(\I')|=0$. In addition, for any $\I''$ with $d(\I,\I'')\leq r$, $|q(\I'')|=N/r$. For any $\I''$ with $d(\I',\I'')\leq r$, $|q(\I'')|\leq r$.

First consider $\I$.  The adversary sets $\mathcal{M}'(\cdot)\equiv N/r$, so $\mathrm{Err}(\mathcal{M}',\I'')=0$ for all $\I''$ in the $r$-neighborhood of $\I$.  Since $\mathcal{M}(\cdot)$ is $(r,c)$-neighborhood optimal, $\mathcal{M}(\I)$ must output $N/r$ deterministically. As $\mathcal{M}(\cdot)$ is $\varepsilon$-DP, $\mathcal{M}(\cdot)$ must output $N/r$ deterministically at all instances.  So its error on $\I'$ is $\mathrm{Err}(\mathcal{M},\I') = N/r$. 

At $\I'$, the adversary sets $\mathcal{M}'(\cdot) \equiv 0$. For any $\I''$ in the $r$-neighborhood of $\I'$, $\mathcal{M}'(\cdot)$ has error at most $\mathrm{Err}(\mathcal{M}',\I'') \leq  r$.  Thus, we conclude that $c\ge {N/r \over r}$, or $cr^2 \ge N$.
\end{proof}

Thus, if one still desires $c=O(1)$, $r$ must be at least $\Omega(\sqrt{N})$.  We also remark that this negative result holds even under relaxed DP, since only substitutions are used when defining the neighborhoods in the proof.  

\section{Practical Performance}

Our analysis in Section~\ref{sec:optimality_analysis} shows that $RS(\cdot)$ is at most a constant-factor larger than $SS(\cdot)$, both of which are $O(1)$-neighborhood optimal. At the same time, $ES(\cdot)$ does not have any optimal guarantee. In this section, we conduct an  experimental study on the actual values of these sensitivities on some simple sub-graph counting queries over real-world graph data.  Note that since the subsequent noise generation process is the same for all three sensitivity measures, it suffices to only compare the sensitivities instead of the $\ell_2$-errors. The experimental results show that the accuracy of $RS(\cdot)$ is very close to that of $SS(\cdot)$ with order-of-magnitude reduction in computational cost. In fact, for most queries, $SS(\cdot)$ is not even known to be polynomially computable.  For these queries, $ES(\cdot)$ is the only known DP mechanism, and $RS(\cdot)$ offers drastic improvement in terms of utility. Finally, it is worth pointing out that our mechanism just requires the evaluation of a number of residual queries, whose results are then combined using a certain formula.  Hence, it can be implemented in a relational DBMS easily (e.g., using PL/SQL or a plug-in). 

\begin{figure}[htbp]
\includegraphics[width=1\textwidth]{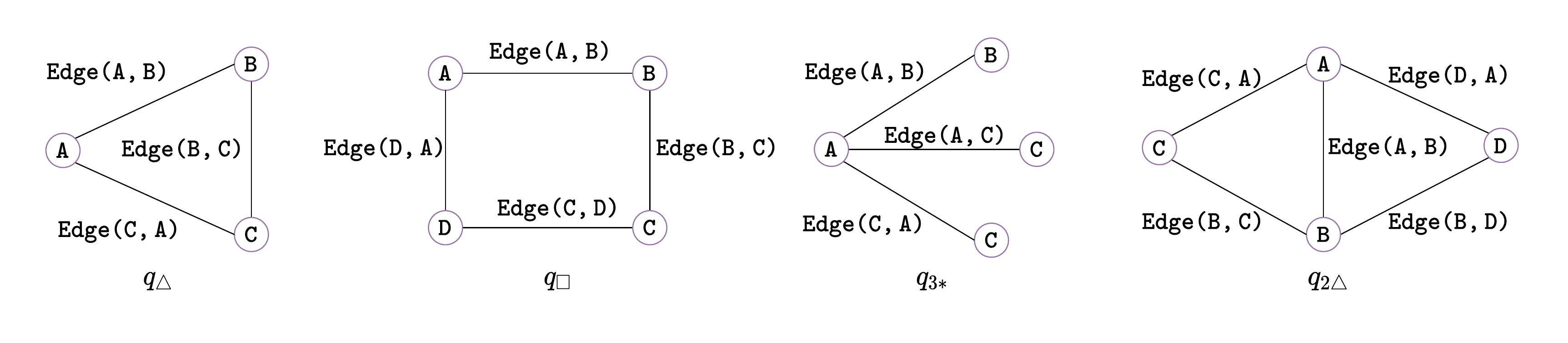}
    \caption{The join structure of queries.}
    \label{fig:queries}
\end{figure}

\subsection{Setup}
\paragraph{Datasets} We use five graph network datasets: $\mathbf{CondMat}$, $\mathbf{AstroPh}$, $\mathbf{HepPh}$, $\mathbf{HepTh}$, and $\mathbf{GrQc}$, which contain $23133$, $18772$, $12008$, $9877$, $5242$ nodes and $186878$, $396100$, $236978$, $51946$, $28980$ edges, respectively.
These five datasets describe the collaboration between the authors on arXiv in Condensed Matter, Physics, High Energy Physics, High Energy Physics Theory and General Relativity categories. The datasets are obtained from SNAP \cite{leskovec2016snap}.  The graphs are directed and we store all edges in a relation $\mathtt{Edge(From,To)}$.

\paragraph{Queries} We experimented with $4$ pattern counting queries as shown in Figure~\ref{fig:queries}.  We also added all inequalities between every two distinct variables.  Since polynomial-time algorithms for computing $SS(\cdot)$ are known only for triangle counting \cite{nissim2007smooth} and $t$-star counting \cite{karwa2011private}, the results on $SS(\cdot)$ are available only on $q_{\triangle}$ and $q_{3*}$  Note that the count returned by the graph pattern counting CQ is actually $3$ times (resp.\ $6$ times) the number of triangles and $3$-stars in the graph, so we need to scale down $SS(\cdot)$ and $ES(\cdot)$ accordingly. 

\begin{table*}[ht]
\resizebox{1\columnwidth}{!}{
\begin{tabular}{c|c|c|c|c|c|c|c}
\hline
\multicolumn{3}{c|}{Dataset}&CondMat&AstroPh&HepPh&HepTh&GrQc
\\
\hline
\hline
\multirow{11}{*}{$q_{\triangle}$}&\multicolumn{2}{c|}{Query result}&1,040,166& 8,108,646&20,150,994&170,034&289,560
\\
\cline{2-8} 
&\multirow{2}{*}{Smooth sensitivity ($SS$)}&Value&489 &  1,050 &1,350 & 102& 183
\\
&&Running Times(s)& 895 & 615 & 261 & 171 & 50.7
\\
\cline{2-8} 
&\multirow{2}{*}{Residual Sensitivity ($RS$)}&Value&493 & 1,054& 1,354& 205& 222
\\
&&Running Times(s)& 6.17 & 24 & 24.7 & 1.37 & 1.05 
\\
\cline{2-8} 
&\multirow{2}{*}{Elastic Sensitivity ($ES$)}&Value&234,361& 763,561 &724,717&12,871&19,927
\\
&&Running Times(s)& 3.5 & 10.8 & 13 & 0.874 & 0.7 
\\
\cline{2-8} 
&\multirow{2}{*}{$RS$ vs $SS$}&Value $RS/SS$&1.01$\times$& $1.00\times$&1.00$\times$&2.01$\times$&1.21$\times$ 
\\
&&Running Times $SS/RS$& 145$\times$ & 25.6$\times$ & 10.6$\times$ & 124$\times$ & 48.3$\times$
\\
\cline{2-8} 
&\multirow{2}{*}{$RS$ vs $ES$}&Value $ES/RS$& 475$\times$& 724$\times$& 535$\times$& 62.8$\times$& 89.8 $\times$
\\
&&Running Times $RS/ES$& 1.76$\times$ & 2.22$\times$ & 1.9$\times$ & 1.57$\times$ & 1.4$\times$  
\\
\hline
\hline
\multirow{11}{*}{$q_{3*}$}&\multicolumn{2}{c|}{Query result}&222,690,360 & 3,274,065,312& 7,661,801,994& 12,590,010& 14,896,428
\\
\cline{2-8} 
&\multirow{2}{*}{Smooth sensitivity ($SS$)}&Value&232,686 & 760,536 & 721,770 & 12,480& 19,440
\\
&&Running Times(s)& 3.09 & 3.08 & 2.47 & 1.92 & 1.52
\\
\cline{2-8} 
&\multirow{2}{*}{Residual Sensitivity ($RS$)}&Value&233,524 & 762,049& 723,244& 12,676& 19,684
\\
&&Running Times(s)& 0.463 & 0.4 & 0.462 & 0.374 & 0.316
\\
\cline{2-8} 
&\multirow{2}{*}{Elastic Sensitivity ($ES$)}&Value&234,361& 763,561& 724,717& 12,871& 19,927
\\
&&Running Times(s)& 0.314 & 0.272 & 0.348 & 0.234 & 0.197
\\
\cline{2-8} 
&\multirow{2}{*}{$RS$ vs $SS$}&Value $RS/SS$&1.00$\times$ &1.00$\times$&1.00$\times$& 1.02$\times$& 1.01$\times$
\\
&&Running Times $SS/RS$& 6.67$\times$ & 7.69$\times$ & 5.34$\times$ & 5.13$\times$ & 4.83$\times$
\\
\cline{2-8} 
&\multirow{2}{*}{$RS$ vs $ES$}&Value $ES/RS$&1.00$\times$ &1.00$\times$ & 1.00$\times$& 1.02$\times$&1.01$\times$ 
\\
&&Running Times $RS/ES$& 1.47$\times$ & 1.47$\times$ & 1.33$\times$ & 1.6$\times$ & 1.6$\times$
\\
\hline
\hline
\multirow{7}{*}{$q_{\square}$}&\multicolumn{2}{c|}{Query result}&12,043,064& 359,332,392& 3,894,935,680& 1,912,648& 8,437,784
\\
\cline{2-8} 
&\multirow{2}{*}{Residual Sensitivity ($RS$)}&Value&12,575 & 72,832& 313,976& 7,089& 8,927
\\
&&Running Times(s)& 41.3 & 296 & 120 & 6.49 & 2.17  
\\
\cline{2-8} 
&\multirow{2}{*}{Elastic Sensitivity ($ES$)}&Value& 87,338,719& 513,622,369& 474,931,535& 1,124,111& 2,165,455
\\
&&Running Times(s)& 2.62$\times$ & 12.5$\times$ & 11.5$\times$ & 0.651$\times$ & 0.388$\times$ 
\\
\cline{2-8} 
&\multirow{2}{*}{$RS$ vs $ES$}&Value $ES/RS$&6,950$\times$& 7,050$\times$& 1,510$\times$& 159$\times$ &243$\times$ 
\\
&&Running Times $RS/ES$& 15.8$\times$ & 23.8$\times$ & 10.4$\times$ & 9.97$\times$ & 5.59$\times$
\\
\hline
\hline
\multirow{7}{*}{$q_{2\triangle}$}&\multicolumn{2}{c|}{Query result}&9,398,600& 289,422,860& 3,747,561,340& 1,716,052& 8,165,996
\\
\cline{2-8} 
&\multirow{2}{*}{Residual Sensitivity ($RS$)}&Value&308,937 & 361,551&515,616 &279,488 & 285,394
\\
&&Running Times(s)& 20.8 & 84.8 & 118 & 4.37 & 4.6
\\
\cline{2-8} 
&\multirow{2}{*}{Elastic Sensitivity ($ES$)}&Value&30,514,062,601& 323,903,424,601& 291,786,363,781& 92,041,951& 220,614,031
\\
&&Running Times(s)& 3.79 & 9.57 & 13.2 & 0.75 & 0.669
\\
\cline{2-8} 
&\multirow{2}{*}{$RS$ vs $ES$}&Value $ES/RS$&98,800$\times$& 896,000$\times$& 566,000$\times$& 329$\times$& 773$\times$
\\
&&Running Times $RS/ES$& 5.49$\times$ & 8.86$\times$ & 8.92$\times$ & 5.77$\times$ & 6.87$\times$ 
\\
\hline
\end{tabular}
}
\caption{Comparison between smooth sensitivity, residual sensitivity and elastic sensitivity when $\beta=0.1$.}
\label{tab:results}
\end{table*}

\begin{figure*}[htbp]
\includegraphics[width=1\textwidth]{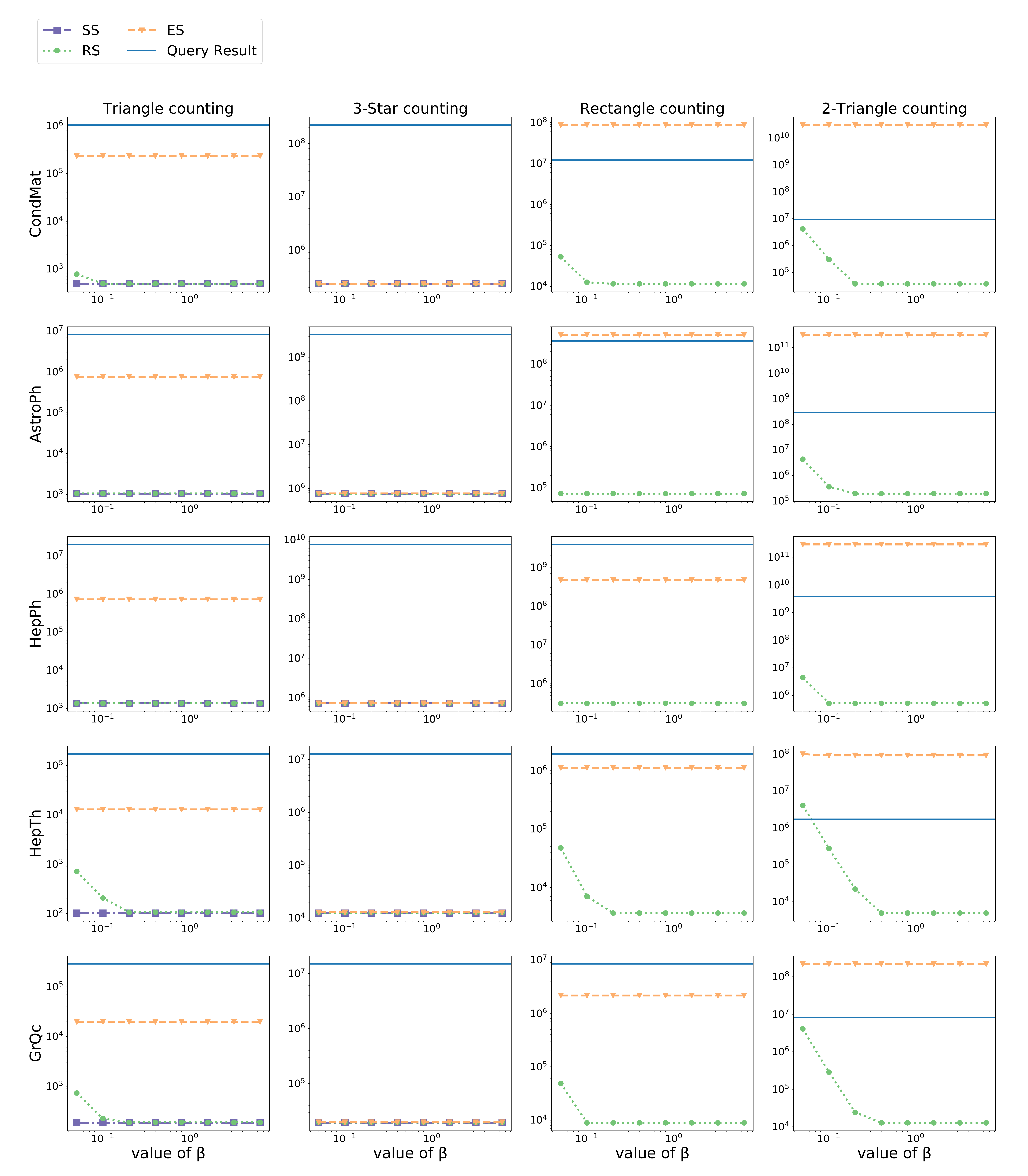}
    \caption{Smooth sensitivity, residual sensitivity and elastic sensitivity with different $\beta$ under different data and queries.}
    \label{fig:results}
\end{figure*}

\subsection{Experimental Results}
The experiments were conducted on a Linux server equipped with a $48$-core 2.2GHz Intel Xeon CPU and 512GB of memory.  For running time, we repeated each query $30$ times and took the average.  Table~\ref{tab:results} gives the results for the setting $\beta=0.1$, which corresponds to $\varepsilon=1$.  

Let's first compare $RS(\cdot)$ and $SS(\cdot)$. We see that $RS(\cdot)$ is only $2\%$ larger than $SS(\cdot)$ in most cases, and the largest difference is $2$ times for $q_{\triangle}$ on the HepTh dataset. This shows that the constant factor derived in Lemma~\ref{lm:compare_RS_SS} is actually quite loose, and the practical utility of $RS(\cdot)$ is close to that of $SS(\cdot)$. On the other hand, $SS(\cdot)$ is much more computational costly: the time for computing $SS(\cdot)$ is $4.83\sim145$ times that of $RS(\cdot)$. 

For the comparison between $RS(\cdot)$ and $ES(\cdot)$, for $q_{\triangle}$, $q_{3*}$ and $q_{2\triangle}$, $ES(\cdot)$ is much larger than $RS(\cdot)$.  On $q_{3*}$, all sensitivity measures are very close.  This is because, for this query, the query result solely depends on the degrees (namely, this is an easy query), and so do all three sensitivities.  Actually, the formulation of $ES(\cdot)$ essentially only makes use of the degree information, which can be verified by the fact that its values on $q_{\triangle}$ and $q_{3*}$ are equal.  On the other hand, $RS(\cdot)$ and $ES(\cdot)$ exploits the actual structure of the graph.

We also tested with different values of $\beta$. The results in Figure~\ref{fig:results} show this does not affect the sensitivity measures much, except for very small values of $\beta$ (i.e., the high privacy regime).

\section{Open Problems}
The negative results on {\sc Median} and non-full CQs hold with respect to the expected $\ell_2$-error.  However, if one uses a high-probability error, then logarithmic-neighborhood optimal mechanisms are possible for the {\sc Median} problem \cite{asi2020instance,huang21mean}.  Whether a similar result can be obtained for non-full CQs remains an interesting open problem.  Another important direction for future work is how to support multiple queries.  Standard DP composition theorems incurs an $O(k)$ or $O(\sqrt{k})$-factor loss when $k$ queries are issued.  For selection queries on a single relation, it has been shown how to reduce this dependency to $O(\log k)$ \cite{hardt2012simple}; how to answer multiple CQs in a way better than standard composition is an interesting yet challenging open question.

\bibliographystyle{plain}
\bibliography{main}
\end{document}